%% file: main.tex
\documentclass[final,1p,times]{elsarticle}

\usepackage{amssymb}

\usepackage{amsmath}
\usepackage{amsthm}
\usepackage{amsfonts}
\usepackage{hyperref}
\usepackage{xspace}
\usepackage{xcolor}

\usepackage{xfrac}
\usepackage{hhline}
\usepackage{pgfplots}
\pgfplotsset{compat=newest,compat/show suggested version=false}
\pgfplotsset{
    legend image with text/.style={
        legend image code/.code={%
            \node[anchor=center] at (0.3cm,0cm) {#1};
        }
    },
}
\makeatletter
\newcommand*{\overtabline}{%
  \noalign{%
    \vskip-.5\dimexpr\ht\@arstrutbox+\dp\@arstrutbox\relax
    \vskip-.2pt\relax
    {\color{red}\hrule height 0.8pt}
    \vskip-.2pt\relax
    \vskip+.5\dimexpr\ht\@arstrutbox+\dp\@arstrutbox\relax
  }%
}
\makeatother
\newtheorem{theorem}{Theorem}
\newtheorem{lemma}[theorem]{Lemma}

\newtheorem{proposition}[theorem]{Proposition}
\newtheorem{definition}[theorem]{Definition}

\newtheorem{example}[theorem]{Example}
\newtheorem{remark}[theorem]{Remark}

\newtheorem{instruction}[theorem]{Instruction}

\usepackage{mathdots,mathtools}
\usepackage[algo2e,boxed,vlined,linesnumberedhidden,
titlenumbered]{algorithm2e}
\SetKw{KwFrom}{from}
\SetKwFor{For}{For}{do}{End for}
\SetKwFor{Forall}{For all}{do}{End for}
\SetKw{KwRet}{Return}
\SetKw{KwErr}{Error}
\SetKw{KwAnd}{and}
\SetKw{KwOr}{or}
\SetKw{KwBreak}{break}
\SetKw{KwContinue}{continue}
\SetKw{KwTrue}{true}
\SetKw{KwFalse}{false}
\SetKw{KwNext}{next}
\SetKwFor{While}{While}{do}{End do}
\SetKwIF{If}{ElseIf}{Else}{If}{then}{Else if}{Else}{End if}

\SetCommentSty{mycommfont}
\usepackage{kbordermatrix}
\renewcommand{\kbldelim}{(}
\renewcommand{\kbrdelim}{)}
\usepackage{ifthen}

\newcommand\mcal{\mathcal}
\newcommand\mbf{\mathbf}
\newcommand\mbb{\mathbb}
\newcommand\msf{\mathsf}
\newcommand\sM{\msf{M}}

\newcommand\f{\mbf{f}}

\newcommand\x{\mbf{x}}

\newcommand\bin{\mbf{b}}
\newcommand\bi{\mbf{i}}

\newcommand\bk{\mbf{k}}

\newcommand\bF{\mbf{F}}

\newcommand\bw{\mbf{w}}
\newcommand\bseq{\bw}
\newcommand\seq{w}

\newcommand\bgamma{\boldsymbol\gamma}

\newcommand\N{\mbb{N}}

\newcommand\K{\mbb{K}}

\newcommand\cG{\mcal{G}}
\newcommand\cH{\mcal{H}}
\newcommand\cB{\mcal{B}}
\newcommand\cK{\mcal{K}}

\newcommand\cT[1][]{
  {\ifthenelse{\equal{#1}{}}{\mcal{T}}{\mcal{T}_{\preceq #1}}}}

\newcommand\ceil[1]{\left\lceil#1\right\rceil}
\newcommand\floor[1]{\left\lfloor#1\right\rfloor}
\newcommand\cro[1]{\left[#1\right]}
\newcommand\pare[1]{\left(#1\right)}
\newcommand\acc[1]{\left\{#1\right\}}
\newcommand\sca[1]{\left\langle#1\right\rangle}
\newcommand\ISSAC{\textsc{ISSAC}\xspace}
\newcommand\gb{Gr\"obner basis\xspace}
\newcommand\gbs{Gr\"obner bases\xspace}
\newcommand\bb{border basis\xspace}
\newcommand\bbs{border bases\xspace}
\newcommand\BM{\textsc{BM}\xspace}
\newcommand\BMS{\textsc{BMS}\xspace}

\newcommand\sFGLM{\textsc{Scalar-FGLM}\xspace}
\newcommand\asFGLM{\textsc{Adaptive Scalar-FGLM}\xspace}

\newcommand\spFGLM{\textsc{Sparse-FGLM}\xspace}
\newcommand\AGbb{\textsc{AGbb}\xspace}
\newcommand\agbb{\textsc{Artinian Gorenstein \bbs}\xspace}
\newcommand\divalgo{\textsc{Polynomial Scalar-FGLM}\xspace}
\newcommand\Adivalgo{\textsc{Adaptive Polynomial Scalar-FGLM}\xspace}
\newcommand\bms{Berlekamp--Massey--Sakata\xspace}
\newcommand\bm{Berlekamp--Massey\xspace}
\newcommand\ie{\mbox{{i.e.}}\xspace}
\newcommand\resp{\mbox{resp.}\xspace}

\newcommand\wrt{\mbox{wrt.}\xspace}

\newcommand\Zero{\textcolor{gray}{0}}
\DeclareMathOperator\DRL{\textsc{drl}}
\DeclareMathOperator\LEX{\textsc{lex}}

\newcommand\adots{\mathinner{%
  \mkern1mu\raise1pt\hbox{.}%
  \mkern2mu\raise4pt\hbox{.}%
  \mkern2mu\raise7pt\vbox{\kern7pt\hbox{.}}\mkern1mu}}

\newcommand\lcm{\textsc{lcm}\xspace}

\DeclareMathOperator\LM{\textsc{lm}}
\DeclareMathOperator\LT{\textsc{lt}}
\DeclareMathOperator\LC{\textsc{lc}}
\DeclareMathOperator\Stabilize{Stabilize}
\DeclareMathOperator\Border{Border}
\DeclareMathOperator\Staircase{Staircase}

\DeclareMathOperator\supp{supp}

\DeclareMathOperator\LCM{\textsc{lcm}}

\DeclareMathOperator\NormalForm{NormalForm}
\DeclareMathOperator\NormalFormRightSide{NormalFormRightSide}
\DeclareMathOperator\NormalFormHigherPart{NormalFormHigherPart}
\newcommand\Card[1]{\#\,#1}

\begin{document}

\begin{frontmatter}



\title{Polynomial-Division-Based Algorithms for Computing Linear
  Recurrence Relations}


\author[label1]{J\'er\'emy Berthomieu\corref{cor1}}
\cortext[cor1]{Laboratoire d'Informatique de Paris~6,
  Sorbonne Universit\'e, bo\^ite courrier~169, 4~place
  Jussieu, F-75252 Paris Cedex~05, France.}
\ead{jeremy.berthomieu@lip6.fr}

\author[label2,label3]{Jean-Charles Faug\`ere}
\address[label1]{Sorbonne Universit\'e, \textsc{CNRS}, \textsc{LIP6},
  F-75005, Paris, France}
\address[label2]{CryptoNext Security}
\address[label3]{Sorbonne Universit\'e, \textsc{CNRS}, \textsc{INRIA},
  \textsc{LIP6},
  F-75005, Paris, France}
\ead{jcf@cryptonext-security.com, jean-charles.faugere@inria.fr}

\begin{abstract}
  \input{0-abstract}
\end{abstract}




\end{frontmatter}


\section{Introduction}
\label{s:intro}
\input{1-intro}

\section{Notation}
\label{s:notation}
\input{2-notation}

\section{From matrices to polynomials}
\label{s:Mat2Pol}
\input{3-representationChange}

\section{A division-based algorithm}
\label{s:division}
\input{4-algo}

\section{An adaptive variant}
\label{s:adaptive}
\input{5-adaptive}

\section{Experiments}
\label{s:bench}
\input{6-bench}

\section*{Acknowledgments}
We thank the anonymous referees for their careful reading and their helpful
comments to improve this paper.
The authors are supported by the joint \textsc{ANR-FWF}
\textsc{ANR-19-CE48-0015} \textsc{ECARP} project, the \textsc{ANR}
grants \textsc{ANR-18-CE33-0011}
\textsc{Sesame} and \textsc{ANR-19-CE40-0018} \textsc{De Rerum Natura}
projects, the \textsc{PGMO} grant \textsc{CAMiSAdo} and the European
Union's Horizon 2020 research and innovation programme under the Marie
Sk\l{}odowska-Curie grant agreement N.~813211 (\textsc{POEMA}).




\bibliographystyle{elsarticle-harv} 
\bibliography{biblio}





\end{document}

%% file: 0-abstract.tex
Sparse polynomial interpolation, sparse linear system solving or
modular rational reconstruction are fundamental problems in Computer
Algebra. They come down to computing linear recurrence relations of a sequence
with the Berlekamp--Massey algorithm.
Likewise, sparse multivariate polynomial interpolation and multidimensional
cyclic code decoding require guessing linear recurrence relations of
a multivariate sequence.

Several algorithms solve this problem.
The so-called 
Berlekamp--Massey--Sakata algorithm (1988) uses
polynomial additions and shifts by a monomial.
The \textsc{Scalar-FGLM} algorithm (2015) relies on linear algebra
operations on a multi-Hankel matrix, a multivariate generalization of
a Hankel matrix.
The Artinian Gorenstein border basis algorithm (2017) uses a
Gram-Schmidt process.

We propose a new algorithm for computing the Gr\"obner basis of the
ideal of 
relations of a sequence based
solely on
multivariate polynomial arithmetic. This algorithm allows us to
both revisit the Berlekamp--Massey--Sakata algorithm through
the use of polynomial divisions and to completely revise
the \textsc{Scalar-FGLM} algorithm without linear algebra operations.

A key observation in the design of this algorithm is to work on the
mirror of the truncated generating series 
allowing us
to
use polynomial arithmetic modulo a 
monomial ideal. It appears to have some similarities with
Pad\'e approximants of this mirror polynomial.

As an addition from the paper published at the \textsc{ISSAC}
conference,
we give an adaptive variant of this algorithm taking into
account the shape of the final Gr\"obner basis gradually as it is
discovered. The main advantage of this algorithm is that its
complexity in terms of operations and sequence queries only depends on
the output Gr\"obner basis.

All these algorithms have been implemented in \textsc{Maple} and we
report on our comparisons.


%% file: 1-intro.tex
The \bm algorithm (\BM), introduced by Berlekamp in 1968~\cite{Berl68}
and Massey in 1969~\cite{Mass69} is a fundamental algorithm in
Coding Theory, \cite{BoseRC1960,Hocquenghem1959}, and Computer
Algebra. It allows one to perform efficiently sparse 
polynomial interpolation, sparse linear system solving or modular rational
reconstruction.

In 1988, Sakata extended the \BM algorithm to dimension $n$. This
algorithm, known as the \bms algorithm (\BMS),
can be used to compute a \gb of the zero-dimensional ideal of the
relations satisfied by a sequence,
\cite{Sakata88,Sakata90,Sakata09}. Analogously to dimension $1$, the
\BMS
algorithm allows one to decode cyclic codes in dimension $n>1$, an
extension of Reed--Solomon's codes. Furthermore, the latest versions
of the \spFGLM
algorithm rely heavily on the efficiency of the \BMS
algorithm to compute the change of ordering of a \gb, \cite{FM11,faugere:hal-00807540}.

\subsection{Related Work}
In dimension $1$, it is well known that the \BM algorithm can be seen
in a matrix form requiring to solve a linear Hankel
system of size $D$, the order of the recurrence, see~\cite{JoMa89}, or
the Levinson--Durbin method, \cite{Levinson47,Wiener49}.
If we let $\sM(D)$ be a cost
function for multiplying two polynomials of degree $D$, for instance
$\sM(D)\in O(D\,\log D\,\log\log D)$, \cite{CantorK91,CooleyT65}, then
solving a linear Hankel system of size $D$ comes down to
performing a truncated extended Euclidean algorithm called on two
polynomials of
degree $D$, \cite{Blackburn1997,BrentGY1980,Dornstetter1987}.
More precisely, it can be done in $O(\sM(D)\,\log D)$
operations.

In~\cite{issac2015,berthomieu:hal-01253934}, the authors present the \sFGLM
algorithm, extending the matrix version of the \BM algorithm for
multidimensional sequences. It consists in computing the relations of the
sequence through the computation of a maximal submatrix of full rank
of a \emph{multi-Hankel} matrix, a multivariate generalization of a Hankel
matrix. Then, it returns the minimal \gb $\cG$ of the ideal of relations
satisfied by the sequence.
These notions are recalled in Section~\ref{s:notation}. 
If we denote by $S$ the staircase defined by
$\cG$ and $T$ the input set of monomials containing $S\cup\LM(\cG)$, then
the complexity of the \sFGLM algorithm is $O((\Card{T})^{\omega})$, where
$2\leq\omega\leq 3$ is the linear algebra exponent. However, we do
not know how to exploit the multi-Hankel structure to improve this
complexity. 

The \agbb algorithm (\AGbb) was presented
in~\cite{Mourrain:2017:FAB:3087604.3087632} for computing a \bb $\cB$
of the ideal of relations. It extends the algorithm
of~\cite{issac2015} using polynomial
arithmetic allowing it to reach the better complexity
$O((\Card{S}+\Card{\cB})\cdot\Card{S}\cdot\Card{T})$ with the above
notation.

Another viewpoint is that computing linear recurrence relations can be
seen as computing Pad\'e
approximants of a truncation of the generating series
$\sum_{i_1,\ldots,i_n\geq 0}w_{i_1,\ldots,i_n}\,x_1^{i_1}\,\cdots\,x_n^{i_n}$.
In~\cite{FitzpatrickF92}, the
authors extend the extended Euclidean algorithm for computing
multivariate Pad\'e approximants. Given a polynomial $P$ and an ideal $B$,
find polynomials $F$ and $C$ such that $P=\frac{F}{C}\bmod B$, where
the leading monomials of $F$ and $C$ satisfy some
constraints.

It is also worth noticing that we now know that both the \BMS and the
\sFGLM algorithms are not
equivalent, see~\cite{berthomieu:hal-01516708}, \ie it is not possible to
tweak one algorithm to mimic the behavior of the other. However, if
the input sequence is linear recurrent and sufficiently many sequence
terms are visited, then both algorithms compute a
\gb of the zero-dimensional ideal of relations.

\subsection{Contributions}
In the whole paper, we assume that the input sets of the \sFGLM
algorithm are the sets of all the monomials less than a given
monomial and that these sets are finite.
In order to improve the complexity of the algorithm, we will use
polynomial arithmetic in all the operations. Even though they are not
equivalent, this reduces the gap between the \BMS and the \sFGLM algorithms and
provides a unified presentation.

In Section~\ref{s:Mat2Pol}, we present the \BM, the \BMS and the \sFGLM
algorithms in a unified polynomial viewpoint. Using the mirror of
the truncated generating series is a key ingredient letting us
perform the computations modulo a specific monomial ideal $B$:
a vector in the kernel of a multi-Hankel matrix is a
polynomial $C$ such that
\begin{equation}\label{eq:intro_lt}
  \LM(P\,C\bmod B)\prec t_C,
\end{equation}
where $P$ is the mirror of the truncated
generating series, $\LM$ denotes the leading monomial and $t_C$ is a
monomial associated to $C$.

One interpretation of this is the computation of
multivariate Pad\'e approximants $\frac{F}{C}$
of $P$ modulo $B$ with different constraints than
in~\cite{FitzpatrickF92} since we require that $\LM(C)$ is in a given set
of terms and $\LM(P\,C \bmod B)$ satisfies equation~\eqref{eq:intro_lt}.

This polynomial point of view allows us to design the
\divalgo algorithm (Algorithm~\ref{algo:divalgo}) in
Section~\ref{s:division} based on
multivariate polynomial divisions. It computes, in a sense, a
generating set of polynomials whose
product with $P$ modulo $B$ must satisfy
equation~\eqref{eq:intro_lt}. If they do not, by polynomial divisions,
we make new ones until finding minimal polynomials satisfying this constraint.
It is worth noticing that in dimension $1$, we recover
the truncated extended Euclidean algorithm applied to the mirror
polynomial of the generating series  of the
input sequence, truncated in degree $D$,
and $x^{D+1}$. All the examples are available on~\cite{web}.

Our main result is Theorem~\ref{th:main}, a simplified version of
which is

\begin{theorem}
  Let $\bseq$ be a sequence, $\prec$ be a total degree monomial ordering and $a$ be
  a monomial. Let us
  assume that the reduced \gb $\cG$ of the ideal of relations of $\bseq$ for
  $\prec$ and its
  staircase $S$ satisfy
  $a\succeq\max(S\cup\LM(\cG))$ and for all $g\preceq a$,
  $s=\max_{\sigma\preceq a}\{\sigma,\
  \sigma\,g\preceq a\}$, we have $\max(S)\preceq s$.  
  Then, the
  \divalgo algorithm applied to $\bseq$, $\prec$ and $a$ terminates and computes a
  minimal \gb of the ideal of relations of
  $\bseq$ for $\prec$ in
  $O\big(\Card{S}\,(\Card{S}+\Card{\cG})\,\Card{\{\sigma,
      \sigma\preceq a\}}\big)$
  operations in the base field.
\end{theorem}
Let us also remark that the complexity bound is based on naive
multivariate polynomial arithmetic and that this algorithm can benefit
from improvements made in this domain.

In applications such as \gbs change of orderings through the \spFGLM algorithm,
sequence queries are costly, \cite{faugere:hal-00807540}. In~\cite{issac2015}, an
adaptive variant of the \sFGLM algorithm was designed aiming to
minimize the number of sequence queries to recover the relations.

In Section~\ref{s:adaptive}, we show how we can transform the
\asFGLM algorithm of~\cite{issac2015} into an algorithm using
polynomial arithmetic. This algorithm is output sensitive and
probabilistic, like the \asFGLM algorithm is. That is, its main
advantage is that its complexity only depends on the sizes of the
computed staircase and \gb.

\begin{theorem}[see Theorem~\ref{th:Adivalgo}]
  Let $\bseq$ be a sequence, $\prec$ be a total degree monomial
  ordering.

  Let us assume that calling the \Adivalgo algorithm on $\bseq$ and
  $\prec$ yields the \gb $\cG$ and its staircase $S$.

  Then, the \Adivalgo algorithm performs at most
  $O((\Card{S}+\Card{\cG})^2\,\Card{2\,S})$ operations in the base
  field and $\Card{2\,(S\cup\LM(\cG))}$ table queries to recover
  $\cG$, where for a set $T$, $2\,T=\{t\,t', t,t'\in T\}$.
\end{theorem}

Finally, in Section~\ref{s:bench}, we compare the
\divalgo 
algorithm with our
implementations of the \BMS,
the \sFGLM and the \AGbb algorithms. Our algorithm performs always
fewer arithmetic operations than the others starting from a certain
size. Even for an example family favorable towards the \BMS
algorithm, our algorithm performs better.

Although we have compared the numbers of arithmetic operations, it
would be beneficial to have an efficient implementation. This would be
the first step into designing a similar efficient algorithm for
computing linear
recurrence relations with polynomial coefficients, extending the
Beckermann--Labahn algorithm~\cite{BeckermannL1994} for computing
multivariate Hermite--Pad\'e approximants.

Amongst the changes from the \ISSAC version of the paper,
\cite{issac2018}, the main additions are a complete redesign of the
\sFGLM algorithm through polynomial arithmetic in
Section~\ref{sss:sfglm} and a full description
of the \Adivalgo algorithm, an adaptive variant of the \divalgo
algorithm using polynomial divisions as well, in
Section~\ref{ss:adivalgo}. 
Generically, one could expect to make one relation
with leading monomial
$m\,x_i^2$ through a division of polynomials related to relations with
leading monomials $m$ and
$m\,x_i$. Yet, the naive approach given
in~\cite[Section~\ref{s:adaptive}]{issac2018} could not do so as it does not
perform any division. The 
\Adivalgo algorithm visits the monomials in the same order as the
\asFGLM algorithm to recover the relations and replaces any linear
algebra computations by polynomial ones, see~\cite{issac2015}.  We
also give the complexity of this algorithm in terms both of the number
of operations and the number of sequence queries.

Furthermore, one of the main obstructions to the design of this
adaptive variant is that at each step, some polynomials are
updated. This update process adds terms supposed to be small with
respect to the ordering. Yet,
their leading terms were not stable during this update process.

Lastly, in Section~\ref{s:Mat2Pol}, we now more clearly define what
$t_C$ should be in Equation~\ref{eq:intro_lt}. This is a key point in
the \divalgo algorithm and a more complete description is available in
Proposition~\ref{prop:lt}.


%% file: 2-notation.tex
We give a brief description of classical notation used in the paper.

\subsection{Sequences and relations}
For $n\geq 1$, we let $\bi=(i_1,\dots,i_n)\in\N^n$ and
for $\x=(x_1,\ldots,x_n)$, we write
$\x^{\bi}=x_1^{i_1}\,\cdots\,x_n^{i_n}$.
\begin{definition}
  Let $\K$ be a field, $\cK\subseteq\N^n$ be finite,
  $\bseq\in\K^{\N^n}$ be a $n$-dimensional sequence with terms in $\K$
  and
  $f=\sum_{\bk\in\cK}\gamma_\bk\,\x^\bk\in\K[\x]$.
  We let $\cro{f}_{\bseq}$, or $\cro{f}$, be the linear combination
  $\sum_{\bk\in\cK}\gamma_\bk \,\seq_\bk$.
  If for all
  $\bi\in\N^n$, $\cro{\x^{\bi}\,f}=0$, then we say that $f$ is the
  \emph{polynomial of the relation induced by
  $\bgamma=(\gamma_\bk)_{\bk\in\cK}\in\K^{\Card{\cK}}$}.
\end{definition}
The main benefit of the $[\,]$ notation resides in the immediate fact
that for any index $\bi$, its \emph{shift by $\x^{\bi}$} is
$\left[\x^\bi\,f\right]=\sum_{\bk\in\cK}
\gamma_\bk\,\seq_{\bk+\bi}$.

\begin{example}\label{ex:binom}
  Let $\bin=\left(\binom{i}{j}\right)_{(i,j)\in\N^2}$ be the sequence of the binomial
  coefficients.
  Then, $x\,y-y-1$ is the polynomial of Pascal's rule:
  \[\forall (i,j)\in\N^2,\, [x^i\,y^j\,(x\,y-y-1)]=
  \bin_{i+1,j+1}-\bin_{i,j+1}-\bin_{i,j}=0.\]
\end{example}

\begin{definition}[\cite{FitzpatrickN90,Sakata88}]~\label{def:lin_rec}
  Let $\bseq=(\seq_{\bi})_{\bi\in\N^n}$ be an $n$-dimensional sequence
  with coefficients in $\K$. The sequence $\bseq$ is \emph{linear
    recurrent} if from a nonzero finite number of initial terms
  $\{\seq_{\bi},\ \bi\in S\}$, and a finite number of relations, without any
  contradiction and without ambiguity, one can compute any term of
  the sequence.
  
  Equivalently, $\bseq$ is linear recurrent if $\{f\in\K[x],\
  \forall\,m\in\K[\x],\cro{m\,f}_{\bseq}=0\}\subseteq\K[x]$,  its ideal of relations, 
  is \emph{zero-dimensional}.
\end{definition}

As the input parameters of the algorithms are the first terms of a
sequence, a \emph{table} shall denote a finite subset of terms of a
sequence.

\subsection{\gbs}
Let $\cT=\{\x^{\bi},\ \bi\in\N^n\}$ be the set of all monomials
in $\K[\x]$.
A monomial ordering $\prec$ on $\K[\x]$ is an order relation
satisfying the following three classical properties:
\begin{enumerate}
\item for all $m\in\cT$, $1\preceq m$;
\item for all $m,m',s\in\cT$, $m\prec m'\Rightarrow m\,s\prec m'\,s$;
\item every subset of $\cT$ has a least element for $\prec$.
\end{enumerate}

For a monomial ordering $\prec$ on $\K[\x]$ and $f\in\K[\x]$, $f\neq 0$, the
\emph{leading monomial} of $f$, denoted $\LM(f)$, is the
greatest monomial in the support of $f$ for $\prec$. The
\emph{leading coefficient} of $f$, denoted $\LC(f)$, is the
nonzero coefficient of $\LM(f)$. The \emph{leading term} of $f$,
$\LT(f)$, is defined as $\LT(f)=\LC(f)\,\LM(f)$. For $f=0$, to
simplify the presentation, we shall define $\LM(f)=\LC(f)=\LT(f)=0$.
For an ideal $I$, we denote
$\LM(I)=\{\LM(f),\ f\in I\}$. Furthermore, we naturally extend $\prec$
to $\cT\cup\{0\}$ with $0\prec 1$.

We recall briefly the definition of a \gb and a staircase.
\begin{definition}\label{def:staircase}
  Let $I$ be a nonzero ideal of $\K[\x]$ and let $\prec$ be a monomial ordering.
  A set $\cG\subseteq I$ is a \emph{\gb} of $I$ if for all $f\in I$,
  there exists $g\in\cG$ such that $\LM(g)|\LM(f)$.

  A \gb $\cG$ of $I$ is \emph{minimal} if for any $g\in\cG$,
  $\langle\cG\setminus\{g\}\rangle\neq I$.

  Furthermore, $\cG$ is \emph{reduced} if for any $g,g'\in\cG$,
  $g\neq g'$ and any monomial $m\in\supp g'$, $\LM(g)\nmid m$.

  The \emph{staircase} of $\cG$
  is defined as
  $S=\Staircase(\cG)=\{s\in\cT,\ \forall\,g\in\cG, \LM(g)\nmid s\}$.
  It is also the canonical basis of $\K[\x]/I$.
\end{definition}

\gb theory allows us to choose any monomial ordering,
among which, we mainly use the
\begin{description}
\item[$\LEX(x_n\prec\cdots\prec x_1)$ ordering] which satisfies
  $\x^{\bi}\prec\x^{\bi'}$ if, and only if, there
  exists $k$, $1\leq k\leq n$, such that for all $\ell<k$,
  $i_{\ell}=i_{\ell}'$ and $i_k<i_k'$,
  see~\cite[Chapter~2, Definition~3]{CoxLOS2015};
\item[$\DRL(x_n\prec\cdots\prec x_1)$ ordering] which satisfies
  $\x^{\bi}\prec\x^{\bi'}$ if, and only if,
  $i_1+\cdots+i_n<i_1'+\cdots+i_n'$ or $i_1+\cdots+i_n=i_1'+\cdots+i_n'$
  and there exists $k$,
  $2\leq k\leq n$, such that for all $\ell>k$, $i_{\ell}=i_{\ell}'$
  and $i_k>i_k'$,
  see~\cite[Chapter~2, Definition~6]{CoxLOS2015}.
\end{description}

However, in the \BMS algorithm, we need to be able to enumerate all
the monomials up to a bound monomial. This forces the user to take an
ordering $\prec$
such that for all $M\in\cT$, the set $\cT[a]=\{m\preceq a,\ m\in\cT\}$
is finite. Such an ordering $\prec$ makes $(\N^n,\prec)$ isomorphic to
$(\N,<)$ as an ordered set. Hence, for a monomial $m$, it makes
sense to speak about the
previous (\resp next) monomial $m^-$ (\resp $m^+$)
for $\prec$. The $\DRL$ ordering
is an example for an ordering on which every term other than $1$ has an
immediate predecessor.

This request excludes for instance the $\LEX$ ordering, and more
generally any elimination ordering. In other words, only weighted
degree ordering, or \emph{weight ordering}, should be used.

Now that a monomial ordering is defined, we can say that
a relation given by a polynomial $f\in\K[\x]$ \emph{fails when
  shifted by $s$} if for all monomials $\sigma\prec s$, $[\sigma\,f]=0$ but
$[s\,f]\neq 0$, see also~\cite{Sakata90,Sakata09}.

\subsection{Multi-Hankel matrices}
A matrix $H\in\K^{m\times n}$ is \emph{Hankel}, if there exists a
sequence $\bseq=(\seq_i)_{i\in\N}$ such that for all
$(i,i')\in\{1,\ldots,m\}\times\{1,\ldots,n\}$, the
coefficient $h_{i,i'}$ lying on the $i$th row and $i'$th column of $H$
satisfies $h_{i,i'}=\seq_{i+i'}$.

In a multivariate setting, we can extend this notion
to \emph{multi-Hankel} matrices. For two sets of monomials $U$ and $T$, we let
$H_{U,T}$ be the multi-Hankel matrix with rows (\resp columns) indexed
with $U$ (\resp $T$) so that the
coefficient of $H_{U,T}$ lying on the row labeled with
$\x^{\bi}\in U$ and column labeled with $\x^{\bi'}\in T$ is
$\seq_{\bi+\bi'}$. 
\begin{example}
  Let $\bseq=(\seq_{i,j,k})_{(i,j,k)\in\N^3}$ be a sequence.
  \begin{enumerate}
  \item For $U=\{1,z,z^2,y,y\,z,y\,z^2\}=\{1,z,z^2\}\,\cup\,y\,\{1,z,z^2\}$ and
    $T=\{1,z,y,y\,z,y^2,y^2\,z\}=\{1,z\}\,\cup\,y\,\{1,z\}\,\cup\,y^2\,\{1,z\}$,
    ordered for $\LEX(z\prec y\prec x)$,
    \[H_{U,T}=\kbordermatrix{
      		&1		&z		&		
                &y              &y\,z		&
		&y^2		&y^2\,z\\
      1		&\seq_{0,0,0}	&\seq_{0,0,1}	&\vrule	
      		&\seq_{0,1,0}	&\seq_{0,1,1}	&\vrule
		&\seq_{0,2,0}	&\seq_{0,2,1}\\
      z		&\seq_{0,0,1}	&\seq_{0,0,2}	&\vrule
		&\seq_{0,1,1}	&\seq_{0,1,2}	&\vrule
		&\seq_{0,2,1}	&\seq_{0,2,2}\\
      z^2	&\seq_{0,0,2}	&\seq_{0,0,3} 	&\vrule
      		&\seq_{0,1,2}	&\seq_{0,1,3}	&\vrule
		&\seq_{0,2,2}	&\seq_{0,2,3}\\
                \hhline{~*{3}{---|}}
      y		&\seq_{0,1,0}	&\seq_{0,1,1} 	&\vrule
      		&\seq_{0,2,0}	&\seq_{0,2,1}	&\vrule
		&\seq_{0,3,0}	&\seq_{0,3,1}\\
      y\,z	&\seq_{0,1,1}	&\seq_{0,1,2} 	&\vrule
      		&\seq_{0,2,1}	&\seq_{0,2,2}	&\vrule
		&\seq_{0,3,1}	&\seq_{0,3,2}\\
      y\,z^2	&\seq_{0,1,2}	&\seq_{0,1,3} 	&\vrule
      		&\seq_{0,2,2}	&\seq_{0,2,3}	&\vrule
		&\seq_{0,3,2}	&\seq_{0,3,3}\\
    }
    \]
    is a $2\times 3$-block-Hankel
    matrix with $3\times 2$-Hankel blocks.
  \item For
    $\tilde{U}=
    U\,\cup\,x\,U\,\cup\,x^2\,U\,\cup\,x^3\,U$ and
    $\tilde{T}=
    T\,\cup\,x\,T\,\cup\,x^2\,T\,\cup\,x^3\,T\,\cup\,x^4\,T$, also
    ordered for $\LEX(z\prec y\prec x)$,
    \[H_{\tilde{U},\tilde{T}}=\kbordermatrix{
      		&T		&		
                &x\,T		&
		&x^2\,T		&
		&x^3\,T		&
		&x^4\,T\\
      U		&H_{U,T}		&\vrule\vrule
      		&H_{U,x\,T}	&\vrule\vrule
      		&H_{U,x^2\,T}	&\vrule\vrule
      		&H_{U,x^3\,T}	&\vrule\vrule
		&H_{U,x^4\,T}\\
                \hhline{~*{4}{---|}}
                \hhline{~*{4}{---|}}
      x\,U	&H_{x\,U,T}	&\vrule\vrule
      		&H_{x\,U,x\,T}	&\vrule\vrule
      		&H_{x\,U,x^2\,T}	&\vrule\vrule
      		&H_{x\,U,x^3\,T}	&\vrule\vrule
		&H_{x\,U,x^4\,T}\\
                \hhline{~*{4}{---|}}
                \hhline{~*{4}{---|}}
      x^2\,U	&H_{x^2\,U,T}	&\vrule\vrule
      		&H_{x^2\,U,x\,T}	&\vrule\vrule
      		&H_{x^2\,U,x^2\,T}	&\vrule\vrule
      		&H_{x^2\,U,x^3\,T}	&\vrule\vrule
		&H_{x^2\,U,x^4\,T}\\
                \hhline{~*{4}{---|}}
                \hhline{~*{4}{---|}}
      x^3\,U	&H_{x^3\,U,T}	&\vrule\vrule
      		&H_{x^3\,U,x\,T}	&\vrule\vrule
      		&H_{x^3\,U,x^2\,T}	&\vrule\vrule
      		&H_{x^3\,U,x^3\,T}	&\vrule\vrule
		&H_{x^3\,U,x^4\,T}\\
    },
    \]
    where $H_{x^i\,U,x^{i'}\,T}=H_{x^{i+i'}\,U,T}$ for any $i,i'$ since
    $H_{x^i\,U,x^{i'}\,T}$ is the same matrix as $H_{U,T}$ where each
    coefficient $w_{0,j,k}$ has been replaced by $w_{i+i',j,k}$. Therefore,
    $H_{\tilde{U},\tilde{T}}$ is a $4\times 5$-\emph{block-Hankel}
    matrix with $6\times 6$-multi-Hankel blocks like $H_{U,T}$.
  \item For $T=\{1,y,x,y^2,x\,y,x^2\}$, ordered for $\DRL(z\prec
    y\prec x)$, $H_{T,T}$ is a multi-Hankel matrix whose structure is
    less clear. It can be
    considered as a block-Hankel matrix with blocks of different sizes,
    noticing that
    $T=\acc{1}\,\cup\,y\,\acc{1,\frac{x}{y}}\,\cup\,y^2\,\acc{1,\frac{x}{y},
      \frac{x^2}{y^2}}$,
    \[H_{T,T}=\kbordermatrix{
      		&1              &
                &y		&x		&
                &y^2            &x\,y		&x^2\\
      1		&\seq_{0,0,0}	&\vrule
      		&\seq_{0,1,0}	&\seq_{1,0,0}	&\vrule
      		&\seq_{0,2,0}	&\seq_{1,1,0} 	&\seq_{2,0,0}\\
                \hhline{~*{3}{---|}}
      y		&\seq_{0,1,0}	&\vrule
      		&\seq_{0,2,0}	&\seq_{1,1,0}	&\vrule
      		&\seq_{0,3,0}	&\seq_{1,2,0}	&\seq_{2,1,0}\\
      x		&\seq_{1,0,0}	&\vrule
      		&\seq_{1,1,0}	&\seq_{2,0,0}	&\vrule
      		&\seq_{1,2,0}	&\seq_{2,1,0}	&\seq_{3,0,0}\\
                \hhline{~*{3}{---|}}
      y^2	&\seq_{0,2,0}	&\vrule
      		&\seq_{0,3,0}	&\seq_{1,2,0}	&\vrule
      		&\seq_{0,4,0}	&\seq_{1,3,0}	&\seq_{2,2,0}\\
      x\,y	&\seq_{1,1,0}	&\vrule
      		&\seq_{1,2,0}	&\seq_{2,1,0}	&\vrule
      		&\seq_{1,3,0}	&\seq_{2,2,0}	&\seq_{3,1,0}\\
      x^2	&\seq_{2,0,0}	&\vrule
      		&\seq_{2,1,0}	&\seq_{3,0,0}	&\vrule
      		&\seq_{2,2,0}	&\seq_{3,1,0}	&\seq_{4,0,0}
    }.
    \]
  \item For
    $U=\{1,z,y,x,z^2,y\,z,x\,z,y^2,x\,y,x^2\}=
    \acc{1}\,\cup\,z\,\acc{1,\frac{y}{z},\frac{x}{z}}\,\cup\,
    z^2\,\acc{1,\frac{y}{z},\frac{x}{z},\frac{y^2}{z^2},\frac{x\,y}{z^2},
      \frac{x^2}{z^2}}$, also
    ordered for $\DRL(z\prec y\prec x)$, the matrix $H_{U,U}$ can be
    seen as a block-matrix like $H_{T,T}$ except each block is a
    multi-Hankel matrix in two variables. In fact, the bottom-right
    block would be the same as $H_{T,T}$ where each coefficient
    $w_{i,j,0}$ is replaced by $w_{i,j,4-i-j}$.
  \end{enumerate}
\end{example}

\subsection{Polynomials associated to multi-Hankel matrices}
\label{ss:pol}
For two sets of terms $T$ and $U$, we let $T+U$ denote their Minkowski sum, \ie
$T+U=\{t\,u,\ t\in T, u\in U\}$, and $2\,T=T+T$.

For a set of terms $T$, we let $M=\LCM(T)$. We
let $P_T$ be the mirror polynomial of the
truncated generating series of a sequence $\bseq$, \ie
\[P_T=\sum_{t\in T}[t]\,\frac{M}{t}.\]
\begin{example}
  Let $\bseq=(\seq_{i,j,k})_{(i,j,k)\in\N^3}$ be a sequence and
  $T=\{1,z,y,x,z^2,y\,z\}$, then $M=x\,y\,z^2$ and
  \begin{align*}
  P_T&=[1]\,x\,y\,z^2+[z]\,x\,y\,z+[y]\,x\,z^2+[x]\,y\,z^2+[z^2]\,x\,y+[y\,z]\,x\,z\\
    &=\seq_{0,0,0}\,x\,y\,z^2+\seq_{0,0,1}\,x\,y\,z+\seq_{0,1,0}\,x\,z^2
      +\seq_{1,0,0}\,y\,z^2+\seq_{0,0,2}\,x\,y+\seq_{0,1,1}\,x\,z.
  \end{align*}
\end{example}
In this paper, we will mostly deal with polynomials
$P_{T+U}$ as there is a strong connection between $H_{U,T}$
and $P_{T+U}$.

Finally, letting $M=\lcm(T+U)=x_1^{D_1}\,\cdots\,x_n^{D_n}$ and $B$ be
the monomial ideal $\big(x_1^{D_1+1},\ldots,x_n^{D_n+1}\big)$,
we will use pairs of multivariate polynomials $R_m=[F_m,C_m]$.
If $m\in B$, then we set $F_m=m$ and $C_m=0$. Otherwise, 
$\LM(C_m)=m$ and $F_m=P_{T+U}\,C_m\bmod B$.


%% file: 3-representationChange.tex
Before detailing the unified polynomial viewpoint, we recall the
linear algebra viewpoint of the \BM, the \BMS and the \sFGLM algorithms.

\subsection{The \BM algorithm}
\label{ss:bm}
Let $\bseq=(\seq_i)_{i\in\N}$ be a one-dimensional table.
Classically, when calling the \BM algorithm, one does not know in
advance the order of the output relation. Therefore, from a matrix
viewpoint, one wants to compute the greatest collection of vectors
\[\pare{
  \begin{smallmatrix}
    \gamma_1\\\vdots\\\gamma_{x^{d-1}}\\1\\0\\0\\\vdots\\0
  \end{smallmatrix}},\pare{
  \begin{smallmatrix}
    0\\\gamma_1\\\vdots\\\gamma_{x^{d-1}}\\1\\0\\\vdots\\0
  \end{smallmatrix}},\ldots,\pare{
  \begin{smallmatrix}
    0\\0\\\vdots\\0\\\gamma_1\\\vdots\\\gamma_{x^{d-1}}\\1
  \end{smallmatrix}}
\]
in the kernel of
$H_{\{1\},\{1,\ldots,x^D\}}=\kbordermatrix{
  	&1	&\cdots	&x^D\\
    1	&\seq_0	&\cdots	&\seq_D
}$, that is $\gamma_1,\ldots,\gamma_{x^{d-1}}\in\K$ such that the relation
$[C_{x^d}]=\seq_d+\sum_{k=0}^{d-1}\gamma_{x^k}\,\seq_k$ and its shifts,
$[x\,C_{x^d}],\ldots,\allowbreak [x^{D-d}\,C_{x^d}]$, are all $0$. 
Equivalently, we look for the least $d$ such that
$H_{\cT[x^{D-d}],\cT[x^d]}\,\pare{
  \begin{smallmatrix}
    \gamma_1\\\vdots\\\gamma_{x^{d-1}}\\1
  \end{smallmatrix}}=0$.

This Hankel matrix-vector product can be extended into
\begin{equation}
  \begin{pmatrix}
    \seq_0	&\cdots		&\seq_{d-1}	&\seq_d\\
    \seq_1	&\cdots		&\seq_d		&\seq_{d+1}\\
    \vdots	&		&\vdots		&\vdots\\
    \seq_{D-d}	&\cdots		&\seq_{D-1}	&\seq_D\\
    \seq_{D-d+1}	&\cdots		&\seq_D		&\Zero\\
    \vdots	&\iddots	&\iddots	&\vdots\\
    \seq_D	&\Zero		&\cdots		&\Zero
  \end{pmatrix}\,
  \begin{pmatrix}
    \gamma_1\\\vdots\\\gamma_{x^{d-1}}\\1
  \end{pmatrix}=
  \begin{pmatrix}
    0\\0\\\vdots\\0\\f_{x^{d-1}}\\\vdots\\f_1
  \end{pmatrix},\label{eq:extendedHankel}
\end{equation}
representing the product of polynomials $P_{\cT[x^D]}=\sum_{i=0}^D
\seq_i\,x^{D-i}$ and $C_{x^d}=x^d+\sum_{k=0}^{d-1}\gamma_{x^k}\,x^k$ modulo
$B=x^{D+1}$. The requirement for $C_{x^d}$ to encode a valid relation is now
that $\LM(F_{x^d})\prec x^d$ with $F_{x^d}=P_{\cT[x^D]}\,C_{x^d} \bmod B$.

This viewpoint gives rise to the following version of the
\BM algorithm: Start with $R_B=[F_B,C_B]=[B,0]$ and $B=x^{D+1}$, and
$R_1=[F_1,C_1]=[P_{\cT[x^D]},1]$.
Compute the quotient $Q$ of the Euclidean division of $F_B=B$ by
$F_1$ and then compute
$R_{\LM(Q)}=R_B-Q\,R_1=[F_B-Q\,F_1,C_B-Q\,C_1]=[F_{\LM(Q)},C_{\LM(Q)}]$.
Repeat
with $R_1$ and $R_{\LM(Q)}$ until reaching
a pair $R_{x^d}=[P_{\cT[x^D]}\,C_{x^d}\bmod
B,C_{x^d}]=\allowbreak[F_{x^d},C_{x^d}]$  with $\LM(C_{x^d})=x^d$
and $\LM(F_{x^d})\prec x^d$. This is in fact the extended Euclidean
algorithm called on $B=x^{D+1}$ and $F_1$ without any computation of the
B\'ezout's cofactors of $x^{D+1}$.
\begin{example}
  Let us consider the Fibonacci table $\bF=(F_i)_{i\in\N}$ with
  $F_0=F_1=1$ and assume $D=5$.
  On the one hand, although the kernel of
  \[H_{\{1\},\{1,x,x^2,x^3,x^4,x^5\}}=\kbordermatrix{
    	&1	&x	&x^2	&x^3	&x^4	&x^5\\
    1	&1	&1	&2	&3	&5	&8
  }\]
  has dimension $5$ and $\pare{
    \begin{smallmatrix}
      -1\\1\\0\\0\\0\\0
    \end{smallmatrix}}$ is in this
  kernel, it corresponds to $[x-1]=0$, its shifted vectors
  $\pare{
    \begin{smallmatrix}
      0\\-1\\1\\0\\0\\0
    \end{smallmatrix}},
  \pare{
    \begin{smallmatrix}
      0\\0\\-1\\1\\0\\0
    \end{smallmatrix}},
  \pare{
    \begin{smallmatrix}
      0\\0\\0\\-1\\1\\0
    \end{smallmatrix}},
  \pare{
    \begin{smallmatrix}
      0\\0\\0\\0\\-1\\1
    \end{smallmatrix}}$ are not, as they correspond to
  $[x^i\,(x-1)]\neq 0$, for $1\leq
  i\leq 4$. However, these vectors
  \[
  \begin{pmatrix}
    -1\\-1\\1\\0\\0\\0
  \end{pmatrix},
  \begin{pmatrix}
    0\\-1\\-1\\1\\0\\0
  \end{pmatrix},
  \begin{pmatrix}
    0\\0\\-1\\-1\\1\\0
  \end{pmatrix},
  \begin{pmatrix}
    0\\0\\0\\-1\\-1\\1
  \end{pmatrix}\]
  are in the kernel and form the greatest collection of shifted
  vectors as such. They correspond to $[x^i\,(x^2-x-1)]=0$,
  for $0\leq i\leq 3$.
  Finally, we have
  \[
  \begin{pmatrix}
    1	&1	&2\\
    1	&2	&3\\
    2	&3	&5\\
    3	&5	&8\\
    5	&8	&\Zero\\
    8	&\Zero	&\Zero
  \end{pmatrix}\,
  \begin{pmatrix}
    -1\\-1\\1
  \end{pmatrix}=
  \begin{pmatrix}
    0\\0\\0\\0\\-13\\-8
  \end{pmatrix},
  \]
  where the gray zeroes ($\Zero$) are due to the matrix extension and
  not the sequence itself.

  On the other hand, $B=x^6$, $R_B=[B,0]$
  and $R_1=[x^5+x^4+2\,x^3+3\,x^2+5\,x+8,1]$. As we can see
  $R_1=[F_1,C_1]$ with $\LM(C_1)=1$ and $\LM(F_1)=x^5\succeq 1$.

  The first step of the
  extended Euclidean algorithm yields
  $R_x=[x^4+x^3+2\,x^2+3\,x-8,x-1]=[F_x,C_x]$ with $\LM(C_x)=x$ and
  $\LM(F_x)=x^4\succeq x$.

  Then, the second step yields
  $R_{x^2}=[-13\,x-8,x^2-x-1]=[F_{x^2},C_{x^2}]$ with $\LM(C_{x^2})=x^2$
  and $\LM(F_{x^2})=x\prec x^2$ so $C_{x^2}$ is a valid relation. We 
  return $C_{x^2}$.
\end{example}

\begin{remark}
  The \BM algorithm always returns a non-zero relation. If no pair
  $R_{x^{\delta}}=[F_{x^{\delta}},C_{x^{\delta}}]$ satisfies the requirements, then
  it will return a pair $R_{x^d}$ with $\LM(C_{x^d})\succ x^D$.
  From a matrix viewpoint, it returns an element of the kernel of the
  empty matrix $H_{\emptyset,\cT[x^d]}$.
\end{remark}
\subsection{Multidimensional extension}
\label{ss:multidim}
In this section, we show how to extend Section~\ref{ss:bm} to
multidimensional sequences. Subsection~\ref{sss:bms} corresponds to the
\BMS algorithm. We shall see that this extension is the closest to the
\BM algorithm. Then, Subsection~\ref{sss:sfglm} corresponds to the
\sFGLM, which, in some way, is more general.
\subsubsection{The \BMS algorithm}
\label{sss:bms}
For a multidimensional table $\bseq=(\seq_{\bi})_{\bi\in\N^n}$, the \BMS
algorithm extends the \BM algorithm by computing vectors in the kernel of a
multi-Hankel matrix
\[H_{\{1\},\cT[a]}=\kbordermatrix{
  	&1	&\cdots	&a^-	&a\\
    1	&[1]	&\cdots	&[a^-]	&[a]
}\]
corresponding to having relations $[C_g]=0$, with
$\LM(C_g)=g$ minimal for the division and for all $t$ such that
$t\,g\preceq a$, $[t\,C_g]=0$ as well.
This also comes down to finding the least
(for the partial order $|$)
monomials $g_1,\ldots,g_r\preceq a$ such that $\dim\ker
H_{\cT[s_k],\cT[g_k]}>0$ with $s_k$ the greatest monomial such that
$s_k\,g_k\preceq a$ for all $k$, $1\leq k\leq r$. Then, each
multi-Hankel matrix-vector product can be extended further as in
equation~\eqref{eq:extendedHankel}, taking the multi-Hankel matrix
$H_{\cT[a],\cT[g_k]}$ and setting to zero any sequence term $[t\,u]$
with $t\,u\notin\cT[a]$.
\begin{equation}
  \begin{pmatrix}
    [1]		&\cdots	&[g_k^-]	&[g_k]\\
    [1^+]	&\cdots	&[1^+\,g_k^-]	&[1^+\,g_k]\\
    \vdots	&	&\vdots		&\vdots\\
    [s_k]	&\cdots	&[s_k\,g_k^-]	&[s_k\,g_k]\\
    [s_k^+]	&\cdots&[s_k^+\,g_k^-]	&\Zero\\
    \vdots	&	&\iddots	&\vdots\\
    [a]		&\Zero	&\cdots		&\Zero
  \end{pmatrix}\,
  \begin{pmatrix}
    \gamma_1\\\vdots\\\gamma_{g_k^-}\\1
  \end{pmatrix}=
  \begin{pmatrix}
    0\\0\\\vdots\\0\\f_{M/s_k^+}\\\vdots\\f_{M/a}
  \end{pmatrix},\label{eq:extendedMultiHankel}
\end{equation}
where $M=\lcm(\cT[a])=x_1^{D_1}\,\cdots\,x_n^{D_n}$. It then represents
the product of polynomials $P_{\cT[a]}=\sum_{t\preceq a}
[t]\,\frac{M}{t}$ and $C_{g_k}=g_k+\sum_{t\prec g_k}\gamma_t\,t$ modulo
$B=(x_1^{D_1+1},\ldots,x_n^{D_n+1})$.
The requirement for $C_{g_k}$ to encode a valid relation is now
that $\LM(F_{g_k})\prec\frac{M}{s_k}$ with $F_{g_k}=P_{\cT[a]}\,C_{g_k} \bmod B$.

Let us notice that $[s_k^+\,g_k^-]$ can also be a $\Zero$
if $s_k^+\,g_k^-\succ a$ and that, more generally, the gray zeroes need not be
diagonally aligned like they are in the univariate case. This is
illustrated by the following example.
\begin{example}
  Let us consider the binomial table $\bin=\pare{\binom{i}{j}}_{(i,j)\in\N^2}$ with
  $\DRL(y\prec x)$ and assume $a=x^2\,y$.
  The kernel of
  \[H_{\{1\},\cT[x^2\,y]}=\kbordermatrix{
    	&1	&y	&x	&y^2	&x\,y	&x^2
        &y^3	&x\,y^2	&x^2\,y\\
    1	&1	&0	&1	&0	&1	&1
    	&0	&0	&2
  }\]
  has clearly dimension $8$. Two vectors in the kernel, $\pare{
    \begin{smallmatrix}
      0\\1\\0\\0\\0\\0\\0\\0\\0
    \end{smallmatrix}}$ and
  $\pare{
    \begin{smallmatrix}
      -1\\\alpha\\1\\0\\0\\0\\0\\0\\0
    \end{smallmatrix}}$, with $\alpha$ any number in $\K$, correspond to the
  independent relations $[y]=0$ and $[x+\alpha\,y-1]=0$.
  However, not all their shifts belong to the kernel. For the former,
  $\pare{\begin{smallmatrix}
      0\\0\\0\\1\\0\\0\\0\\0\\0
    \end{smallmatrix}}$, corresponding to $[y^2]=0$, belongs to the
  kernel but $\pare{\begin{smallmatrix}
      0\\0\\0\\0\\1\\0\\0\\0\\0
    \end{smallmatrix}}$ does not, as $[x\,y]\neq 0$. For the latter,
  $\pare{\begin{smallmatrix}
      0\\-1\\0\\\alpha\\1\\0\\0\\0\\0
    \end{smallmatrix}}$ does not belong to the kernel, as
  $[x\,y+\alpha\,y^2-y]\neq 0$, whatever $\alpha$ is.
  
  Therefore, the vectors in the kernel that we seek must correspond
  to relations $[C_g]=0$ with $g\in\{y^2,x\,y,x^2,\ldots,x^2\,y\}$.
  
  In fact, as $[m\,y^2]=0$ for $1\preceq m\preceq x$, the vectors
  $\pare{
    \begin{smallmatrix}
      0\\0\\0\\1\\0\\0\\0\\0\\0
    \end{smallmatrix}},\pare{
    \begin{smallmatrix}
      0\\0\\0\\0\\0\\0\\1\\0\\0
    \end{smallmatrix}},\pare{
    \begin{smallmatrix}
      0\\0\\0\\0\\0\\0\\0\\1\\0
    \end{smallmatrix}}$,
  are in the kernel and fulfill the requirements. We can indeed notice that
  \[H_{\cT[x],\cT[y^2]}\,
  \begin{pmatrix}
    0\\0\\0\\1
  \end{pmatrix}=
  \kbordermatrix{
    	&1	&y	&x	&y^2\\
    1	&1	&0	&1	&0\\
    y	&0	&0	&1	&0\\
    x	&1	&1	&1	&0
  }\,
  \begin{pmatrix}
    0\\0\\0\\1
  \end{pmatrix}=
  \begin{pmatrix}
    0\\0\\0
  \end{pmatrix},\quad
  \kbordermatrix{
    	&1	&y	&x	&y^2\\
    1	&1	&0	&1	&0\\
    y	&0	&0	&1	&0\\
    x	&1	&1	&1	&0\\
    y^2	&0	&0	&0	&\Zero\\
    x\,y&1	&0	&2	&\Zero\\
    x^2	&1	&2	&\Zero	&\Zero\\
    y^3	&0	&\Zero	&\Zero	&\Zero\\
    x\,y^2&0	&\Zero	&\Zero	&\Zero\\
    x^2\,y&2	&\Zero	&\Zero	&\Zero
  }\,
  \begin{pmatrix}
    0\\0\\0\\1
  \end{pmatrix}=
  \begin{pmatrix}
    0\\0\\0\\0\\0\\0\\0\\0\\0
  \end{pmatrix},
  \]
  where the gray zeroes ($\Zero$) are due to the matrix extension and
  not the binomial sequence itself. From the polynomial point of view,
  \begin{align*}
    F_{y^2}	&=P_{\cT[x^2\,y]}\,C_{y^2}\bmod B\\
    		&=(x^2\,y^3+x\,y^3+x\,y^2+y^3+2\,y^2)\,y^2\bmod
                  (x^3,y^4)\\
    		&=0.
  \end{align*}

  Likewise, the
  vectors
  $\pare{
    \begin{smallmatrix}
      -(1+\alpha)\\\beta\\\alpha\\0\\1\\0\\0\\0\\0
    \end{smallmatrix}},\pare{
    \begin{smallmatrix}
      0\\-(1+\alpha)\\0\\\beta\\\alpha\\0\\0\\1\\0
    \end{smallmatrix}},\pare{
    \begin{smallmatrix}
      0\\0\\-(1+\alpha)\\0\\\beta\\\alpha\\0\\0\\1
    \end{smallmatrix}}$,
  are in the kernel for $\alpha=0$ and $\beta=-1$, as they
  correspond to $[m\,(x\,y+\alpha\,x+\beta\,y-(1+\alpha))]=[m\,(x\,y-y-1)]=0$, for
  $1\preceq m\preceq x$. Furthermore,
  \[\kbordermatrix{
    	&1	&y	&x	&y^2	&x\,y\\
    1	&1	&0	&1	&0	&1\\
    y	&0	&0	&1	&0	&0\\
    x	&1	&1	&1	&0	&2\\
    y^2	&0	&0	&0	&\Zero	&\Zero\\
    x\,y&1	&0	&2	&\Zero	&\Zero\\
    x^2	&1	&2	&\Zero	&\Zero	&\Zero\\
    y^3	&0	&\Zero	&\Zero	&\Zero	&\Zero\\
    x\,y^2&0	&\Zero	&\Zero	&\Zero	&\Zero\\
    x^2\,y&2	&\Zero	&\Zero	&\Zero	&\Zero
  }\,
  \begin{pmatrix}
    -1\\-1\\0\\0\\1
  \end{pmatrix}=
  \begin{pmatrix}
    0\\0\\0\\0\\-1\\-3\\0\\0\\-2
  \end{pmatrix},
  \]
  From the polynomial point of view, $F_{x\,y}=P_{\cT[x^2\,y]}\,\,(x\,y-y-1)\bmod
  (x^3,y^4)=-x\,y^2-3\,y^3-2\,y^2$.

  Finally, the
  vectors
  $\pare{
    \begin{smallmatrix}
      -(1+\alpha)\\\beta\\\alpha\\0\\0\\1\\0\\0\\0
    \end{smallmatrix}},\pare{
    \begin{smallmatrix}
      0\\-(1+\alpha)\\0\\\beta\\\alpha\\0\\0\\0\\1
    \end{smallmatrix}}$,
  are in the kernel for $\alpha=-2$ and $\beta=0$, as they
  correspond to $[m\,(x^2+\alpha\,x+\beta\,y-(1+\alpha))]=[m\,(x^2-2\,x+1)]=0$, for
  $1\preceq m\preceq y$. Furthermore,
  \[\kbordermatrix{
    	&1	&y	&x	&y^2	&x\,y	&x^2\\
    1	&1	&0	&1	&0	&1	&1\\
    y	&0	&0	&1	&0	&0	&2\\
    x	&1	&1	&1	&0	&2	&\Zero\\
    y^2	&0	&0	&0	&\Zero	&\Zero	&\Zero\\
    x\,y&1	&0	&2	&\Zero	&\Zero	&\Zero\\
    x^2	&1	&2	&\Zero	&\Zero	&\Zero	&\Zero\\
    y^3	&0	&\Zero	&\Zero	&\Zero	&\Zero	&\Zero\\
    x\,y^2&0	&\Zero	&\Zero	&\Zero	&\Zero	&\Zero\\
    x^2\,y&2	&\Zero	&\Zero	&\Zero	&\Zero	&\Zero
  }\,
  \begin{pmatrix}
    1\\0\\-2\\0\\0\\1
  \end{pmatrix}=
  \begin{pmatrix}
    0\\0\\-1\\0\\-3\\1\\0\\0\\2
  \end{pmatrix}.
  \]
  From the polynomial point of view, $F_{x^2}=P_{\cT[x^2\,y]}\,\,(x^2-2\,x+1)\bmod
  (x^3,y^4)=-x\,y^3-3\,x\,y^2+y^3+2\,y^2$.
  
  The \BMS algorithm will return the three relations $C_{y^2}=y^2$,
  $C_{x\,y}=x\,y-y-1$ and $C_{x^2}=x^2-2\,x+1$.
\end{example}
\begin{remark}
  As for the \BM algorithm, the \BMS algorithm will always return a
  relation $C_g$ with $\LM(C_g)=g$ a pure power in each
  variable. Therefore, it can return $C_g$ with $g\succ a$,
  corresponding to a vector in the kernel of the empty
  matrix $H_{\emptyset,\cT[g]}$.
\end{remark}

\subsubsection{The \sFGLM algorithm}
\label{sss:sfglm}
The \sFGLM algorithm aims to compute vectors in the
kernel of a more general multi-Hankel matrix $H_{U,T}$, with $T$ and $U$ two
ordered sets of monomials. The kernel vectors that we wish to compute are those
corresponding to relations $[u C_g]=0$ with
$\LM(C_g)=g$ and $u\in U$, such that for all $t\in\cT$, if $t\,g\in T$, then
$[u\,t\,C_g]=0$. In other words, if the vector corresponding to $C_g$ is
in the kernel, then for all $t$ with $t\,g\in T$, so is the vector
corresponding to $t\,C_g$.

Let us assume that both sets of terms $T$ and
$U$ satisfy $T=\cT[a]$ and $U=\cT[b]$. This allows us to encompass
both the \BMS algorithm and the \sFGLM algorithm.

The goal is to extend the multi-Hankel matrix-vector product
\begin{equation}
\kbordermatrix{
	&1	&\cdots	&a^-		&a\\
1	&[1]	&\cdots	&[a^-]		&[a]\\
\vdots	&\vdots	&	&\vdots		&\vdots\\
b^-	&[b^-]&\cdots	&[a^-\,b^-]	&[a\,b^-]\\
b	&[b]	&\cdots	&[a^-\,b]	&[a\,b]
}\,
\begin{pmatrix}
  \gamma_1\\\vdots\\\gamma_{g^-}\\1\\0\\\vdots\\0
\end{pmatrix}=
\begin{pmatrix}
0\\\vdots\\0\\0
\end{pmatrix}\label{eq:multiHankel}
\end{equation}
in a similar fashion as in equations~\eqref{eq:extendedHankel}
and~\eqref{eq:extendedMultiHankel} with as many rows as possible
and setting any table term 
$[t\,u]$ to zero whenever $t\,u\succ a\,b$.

For $C_g$ a relation, we want to build a multi-Hankel matrix $H_{U,T}$
whose kernel contains the vector 
corresponding to $C_g$ if, and only if, the vectors
corresponding to $t\,C_g$ are in the kernel of
$H_{\cT[b],\cT[a]}$. This is achieved by choosing $T=\cT[g]=\{1,\ldots,g\}$
and picking, for each $t$, the rows labeled with
$\{t,\ldots,b\,t\}=\cT[b]+\{t\}$. Thus, the set
of all rows is $U=\cT[b]+\cT[s]$ where
$s$ is the largest monomial such that $s\,g\preceq a$.

Now, we can expand this matrix by adding rows up to $a\,b$ and
setting table terms $[t\,u]$ to $0$ whenever $t\,u\succ a\,b$. Yet, since
the first rows are in $\cT[b]+\cT[s]$ and the set of rows should be
stable by division, we remove from $\cT[a]+\cT[b]$ any multiple of a
monomial in $\cT[b\,s]\setminus(\cT[b]+\cT[s])$.

It remains to make the link between this matrix-vector product and the product of
the two polynomials $P_{\cT[a]+\cT[b]}$ and $C_g$ modulo
$B=(x_1^{D_1+1},\ldots,x_n^{D_n+1})$.
Note that since $\cT[b]+\cT[s]\subseteq\cT[b\,s]$ but may not
be equal to $\cT[b\,s]$, the leading monomial of
$F_g=P_{\cT[a]+\cT[b]}\,C_g\bmod B$ is $\frac{M}{u}$ with
$u$ a multiple of a monomial that may not be in $\cT[b]+\cT[s]$. We let
$\tilde{F}_g$ be the
same polynomial as $F_g$ where we set to zero any monomial
$\frac{M}{u}$ in $F_g$ with
$u$ a multiple of a monomial in $\cT[b\,s]\setminus(\cT[b]+\cT[s])$. Since the
monomials $\frac{M}{u}$ with $u\in\cT[b]+\cT[s]$ are now the largest
possible monomials in $\tilde{F}_g$, $C_g$ is a valid
relation if, and only if, $\LM(\tilde{F}_g)\prec\frac{M}{b\,s}$.

\begin{proposition}
  \label{prop:lt}
  Let $T=\cT[a]$ and $U=\cT[b]$ be finite sets of monomials in $\K[\x]$, let
  $M=\lcm(T+U)=x_1^{D_1}\,\cdots\,x_n^{D_n}$ and $B=(x_1^{D_1+1},\ldots,x_n^{D_n+1})$.

  Let $C_g$ be a polynomial with support in $T$ and with leading
  monomial $g$ and let $s$ be the greatest monomial such that
  $s\,g\preceq a$.

  Let
  $\cG_s$ be a minimal set of monomials generating the sets of monomials
  less than $b\,s$ but not in $\cT[b]+\cT[s]$, \ie $\cG_s$ is a reduced \gb of
  the monomial ideal generated by $\cT[b\,s]\setminus(\cT[b]+\cT[s])$.

  Let
  $\tilde{F}_g$ be the polynomial obtained by setting to
  zero all the coefficients of monomials
  $\frac{M}{u}$ of
  $F_g=P_{T+U}\,C_g\bmod B
  =\sum_{\tau\in(T+U)}f_{g,\tau}\,\frac{M}{\tau}$,
  with $u\in\sca{\cG_s}$. That is,
  $\tilde{F}_g=\sum_{\tau\in(T+U)}\tilde{f}_{g,\tau}\,\frac{M}{\tau}$, where
  $\NormalForm\pare{\sum_{\tau\in(T+U)}f_{g,\tau}\,\tau,\cG_s}
  =\sum_{\tau\in(T+U)}\tilde{f}_{g,\tau}\,\tau$.

  Then, $[u\,t\,C_g]=0$ for all $u\in U$ and all
  $t\in\cT[s]$ if, and only if, $\LM(\tilde{F}_g)\prec\frac{M}{b\,s}$.
\end{proposition}

\begin{example}\label{ex:sFGLM_bin}
  We still consider the binomial table $\bin$ with
  $\DRL(y\prec x)$. We let $a=b=x\,y^2$ so that $T=U=\cT[x\,y^2]$
  \[
  H_{U,T}=
  \kbordermatrix{
    	&1	&y	&x	&y^2	&x\,y	&x^2
        &y^3	&x\,y^2\\
    1	&1	&0	&1	&0	&1	&1
    	&0	&0\\
    y	&0	&0	&1	&0	&0	&2
    	&0	&0\\
    x	&1	&1	&1	&0	&2	&1
    	&0	&1\\
    y^2	&0	&0	&0	&0	&0	&1
	&0	&0\\
    x\,y&1	&0	&2	&0	&1	&3
	&0	&0\\
    x^2	&1	&2	&1	&1	&3	&1
	&0	&3\\
    y^3	&0	&0	&0	&0	&0	&0
    	&0	&0\\
    x\,y^2&0	&0	&1	&0	&0	&3
	&0	&0\\
  }.\]

  The computation of the kernel of this matrix yields the vectors $\pare{
    \begin{smallmatrix}
      -1\\-1\\0\\0\\1\\0\\0\\0\\
    \end{smallmatrix}},\pare{
    \begin{smallmatrix}
      0\\-1\\0\\-1\\0\\0\\0\\1\\
    \end{smallmatrix}}$ corresponding to
  $[u\,(x\,y-y-1)]=[u\,y\,(x\,y-y-1)]=0$ for all
  $u\in U$ and the vector $\pare{
    \begin{smallmatrix}
      0\\0\\0\\0\\0\\0\\1\\0
    \end{smallmatrix}}$ corresponding to $[u\,y^3]=0$ for all $u\in U$.
  
  For $g=x\,y$, since $g\,y=x\,y^2=a$, we have $s=y$. Then,
  $\cT[b]+\cT[s]=\cT[x\,y^2]+\cT[y]=\cT[x\,y^3]\setminus\{x^3\}$.
  Thus,
  we do not add any row with label multiple of $x^3$ in
  the extended matrix-vector product. That is, the set of rows of the
  extended matrix is
  $(\cT[x\,y^2]+\cT[x\,y^2])\setminus\{x^3,x^3\,y,x^4,x^3\,y^2\}$:
  \[
  \kbordermatrix{
    		&1	&y	&x	&y^2	&x\,y\\
    1		&1	&0	&1	&0	&1\\
    y		&0	&0	&1	&0	&0\\
    x		&1	&1	&1	&0	&2\\
    y^2		&0	&0	&0	&0	&0\\
    x\,y	&1	&0	&2	&0	&1\\
    x^2		&1	&2	&1	&1	&3\\
    y^3		&0	&0	&0	&0	&0\\
    x\,y^2	&0	&0	&1	&0	&0\\
    x^2\,y	&2	&1	&3	&0	&3\\
    x^3    	&1	&3	&1	&3	&\Zero\\
    \overtabline
    y^4		&0	&0	&0	&0	&0\\
    x\,y^3	&0	&0	&0	&0	&0\\
    x^2\,y^2	&1	&0	&3	&0	&\Zero\\
    x^3\,y    	&3	&3	&\Zero	&\Zero	&\Zero\\
    \overtabline
    x^4    	&1	&\Zero	&\Zero	&\Zero	&\Zero\\
    \overtabline
    y^5		&0	&0	&0	&\Zero	&\Zero\\
    x\,y^4	&0	&0	&0	&\Zero	&\Zero\\
    x^2\,y^3	&0	&0	&\Zero	&\Zero	&\Zero\\
    x^3\,y^2    &3	&\Zero	&\Zero	&\Zero  &\Zero\\
    \overtabline
    y^6		&0	&\Zero	&\Zero	&\Zero	&\Zero\\
    x\,y^5	&0	&\Zero	&\Zero	&\Zero	&\Zero\\
    x^2\,y^4	&0	&\Zero	&\Zero	&\Zero	&\Zero\\
  }\,
  \begin{pmatrix}
    -1\\-1\\0\\0\\1
  \end{pmatrix}=
  \begin{pmatrix}
    0\\0\\0\\0\\0\\0\\0\\0\\0\\-4\\\overtabline
    0\\0\\-1\\-6\\\overtabline
    -1\\\overtabline
    0\\0\\0\\-3\\\overtabline
    0\\0\\0\\
  \end{pmatrix}.
  \]
  From the polynomial viewpoint, $F_{x\,y}=P_{T+U}\,(x\,y-y-1)\bmod
  (x^5,y^7)=-4\,x\,y^6-x^2\,y^4-6\,x\,y^5-y^6-3\,x\,y^4$ so that
  $\tilde{F}_{x\,y}=-x^2\,y^4=\sum_{\tau\in(T+U)}\tilde{f}_{g,\tau}\,\frac{M}{\tau}$.
  The vector $\pare{
    \begin{smallmatrix}
      -1\\-1\\0\\0\\1
    \end{smallmatrix}}$ being in the kernel of
  $H_{\cT[x\,y^2]+\cT[y],\cT[x\,y]}$ is then equivalent to asking that
  $\LM(\tilde{F}_{x\,y})\prec\frac{M}{b\,s}=\frac{x^4\,y^6}{x\,y^3}=x^3\,y^3$. Notice
  that this condition is not satisfied by $F_{x\,y}$ since $\LM(F_{x\,y})=x\,y^6$.

  For $g=y^3$, since $a=x\,y^2$, we have $s=1$. Then,
  $\cT[g]+\cT[s]=\cT[y^3]+\cT[1]=\cT[y^3]$, thus we add all
  the rows from $\cT[x\,y^2]+\cT[x\,y^2]$.
  \[
  \kbordermatrix{
    		&1	&y	&x	&y^2	&x\,y	&x^2	&y^3\\
    1		&1	&0	&1	&0	&1	&1	&0\\
    y		&0	&0	&1	&0	&0	&2	&0\\
    x		&1	&1	&1	&0	&2	&1	&0\\
    y^2		&0	&0	&0	&0	&0	&1	&0\\
    x\,y	&1	&0	&2	&0	&1	&3	&0\\
    x^2		&1	&2	&1	&1	&3	&1	&0\\
    y^3		&0	&0	&0	&0	&0	&0	&0\\
    x\,y^2	&0	&0	&1	&0	&0	&3	&0\\
    x^2\,y	&2	&1	&3	&0	&3	&\Zero	&0\\
    x^3		&1	&3	&1	&3	&\Zero	&\Zero	&\Zero\\
    y^4		&0	&0	&0	&0	&0	&0	&\Zero\\
    x\,y^3	&0	&0	&0	&0	&0	&\Zero	&\Zero\\
    x^2\,y^2	&1	&0	&3	&0	&\Zero	&\Zero	&\Zero\\
    x^3\,y	&3	&3	&\Zero	&\Zero	&\Zero	&\Zero	&\Zero\\
    x^4		&1	&\Zero	&\Zero	&\Zero	&\Zero	&\Zero	&\Zero\\
    y^5		&0	&0	&0	&\Zero	&\Zero	&\Zero	&\Zero\\
    x\,y^4	&0	&0	&0	&\Zero	&\Zero	&\Zero	&\Zero\\
    x^2\,y^3	&0	&0	&\Zero	&\Zero	&\Zero	&\Zero	&\Zero\\
    x^3\,y^2	&3	&\Zero	&\Zero	&\Zero	&\Zero	&\Zero	&\Zero\\
    y^6		&0	&\Zero	&\Zero	&\Zero	&\Zero	&\Zero	&\Zero\\
    x\,y^5	&0	&\Zero	&\Zero	&\Zero	&\Zero	&\Zero	&\Zero\\
    x^2\,y^4	&0	&\Zero	&\Zero	&\Zero	&\Zero	&\Zero	&\Zero\\
  }\,
  \begin{pmatrix}
    0\\0\\0\\0\\0\\0\\1
  \end{pmatrix}=
  \begin{pmatrix}
    0\\0\\0\\0\\0\\0\\0\\0\\0\\0\\0\\0\\0\\0\\0\\0\\0\\0\\0\\0\\0\\0
  \end{pmatrix}.
  \]
  The polynomial $F_{y^3}=P_{T+U}\,y^3\bmod
  (x^5,y^7)=0$, so that
  $\tilde{F}_{y^3}=0$. Trivially, $\LM(\tilde{F}_{y^3})$ satisfies any
  constraint on its leading monomial, in particular
  $\LM(\tilde{F}_{y^3})\prec\frac{M}{b\,s}=\frac{x^4\,y^6}{x\,y^2}=x^2\,y^4$.
\end{example}


%% file: 4-algo.tex
The goal is now to design an algorithm based on polynomial division to
determine all the $C_g$ for $|$-minimal $g$ such that
$\LM(\tilde{F}_g)$ is small enough, where $F_g=P_{T+U}\,C_g\bmod B$
and $\tilde{F}_g$ is obtained from $F_g$ as in
Proposition~\ref{prop:lt}.

We start with two sets of terms $T=\cT[a]$ and $U=\cT[b]$ so that
$M=x_1^{D_1}\,\cdots\,x_n^{D_n}=\lcm(T+U)$.
We initialize
$B=(B_1,\ldots,B_n)=\big(x_1^{D_1+1},\ldots,x_n^{D_n+1}\big)$,
$R_{B_1}=[B_1,0],\ldots,R_{B_n}=[B_n,0]$ and
$R_1=[P_{T+U},1]$.

For any monomial $g$ not in the ideal spanned by $B$ and
$R_g=[F_g,C_g]=[P_{T+U}\,C_g\bmod B,C_g]$, by Proposition~\ref{prop:lt}, $C_g$ is a
valid relation if $g\in T$ and $\LM(\tilde{F}_g)\prec\frac{M}{b\,s}$ with
$s=\max\{\sigma,\ \sigma\,g\preceq a\}$.
To go along with the fact that the
\BMS algorithm always returns a relation $C_g$ with $g=\LM(C_g)$ a
pure power of a variable, if $g\not\in T$, then $C_g$ will
automatically be
considered a valid relation as well.

From a failing relation $C_m$, we get that $m$
is in the staircase of the \gb of relations. Thus, each time a built
relation is not valid, we update the staircase of the ideal of
relations. At each step, we
know a staircase $S$ which is a subset of the output and target
staircase. Equivalently, we know the set
$\cH=\min_|\{h\in\cT\setminus S\}$ which are the leading terms of the
candidate relations.

The algorithm  uses the following subroutines
\begin{description}
\item[$\NormalForm(R_{m'},{[R_m,R_{B_1},\ldots,R_{B_n},R_{t_1},\ldots,R_{t_r}]})$,]
  for computing the \emph{normal form} of
  $[F_{m'},C_{m'}]$ \wrt the list
  $R_m,R_{B_1},\ldots,R_{B_n},\allowbreak R_{t_1},\ldots,R_{t_r}$
  with $\LM(F_{t_1})\succ\cdots\succ\LM(F_{t_r})$. To do so, it computes first
  $Q_m,Q_{B_1},\ldots,Q_{B_n},\allowbreak Q_{t_1},\ldots,Q_{t_r}$ the quotients of
  the division of $F_{m'}$ by the list of polynomials
  $[F_m,B_1,\ldots,B_n,F_{t_1},\ldots,F_{t_r}]$ and then return
  $R_h=R_{m'}-Q_m\,R_m-Q_{B_1}\,R_{B_1}-\cdots-Q_{B_n}\,R_{B_n}
  -Q_{t_1}\,R_{t_1}-\cdots-Q_{t_r}\,R_{t_r}$.
\item[$\Stabilize(S)$,] for computing the true staircase containing
  $S$, \ie all the divisors of terms in $S$.
\item[$\Border(S)$,] for computing the least terms for $|$ outside of $S$.
\end{description}

For $h\in\cH$, we now want to build $R_h$ with the least $\LM(\tilde{F}_h)$.
\begin{instruction}\label{ins:newpair}Pick a failing
  pair $R_m=[F_m,C_m]$ with $h=q\,m$ and $m$ the largest for $\prec$,
  \begin{enumerate}
  \item if there exists another failing pair $R_{m'}=[F_{m'},C_{m'}]$
    such that $\LM(F_{m'})=q\,\LM(F_m)$, then compute $R_h$ as the
    $\NormalForm (R_{m'},[R_m,R_{B_1},\ldots,R_{B_n},R_{t_1},\ldots,R_{t_n}])$ where
    $C_{t_1},\ldots,C_{t_r}$ are failing relations and
    $\LM(F_{t_1})\succ\cdots\succ\LM(F_{t_r})$. 
  \item\label{enum:oob} otherwise, compute $R_h$ as
    $\NormalForm (q\,R_m,[R_{B_1},\ldots,R_{B_n},R_{t_1},\ldots,R_{t_n}])$.
  \end{enumerate}
\end{instruction}

\begin{remark}
  If $q\,\LM(F_m)$ is in the ideal spanned by $B$,
  then case~\ref{enum:oob} of Instruction~\ref{ins:newpair} is
  equivalent to computing the normal form of 
  $[q\,\LM(F_m),0]$ \wrt $[R_m,R_{B_1},\ldots,R_{B_n},\allowbreak
  R_{t_1},\ldots,\allowbreak R_{t_r}]$. In fact, unless the table is
  $0$, at
  the start, $R_1=[P_{T+U},1]$ must fail when shifted by a monomial
  $s=x_1^{i_1}\,\cdots\,x_n^{i_n}$ and we have to make new pairs
  $R_{x_1^{i_1+1}},\ldots,R_{x_n^{i_n+1}}$. Since
  $\LM(P_{T+U})=\frac{M}{s}$, then these pairs can be computed as
  the normal forms of
  $[x_k^{i_k+1}\,\frac{M}{s},0]=[B_k,0]\,M'$, with $M'\in\cT$,
  \wrt the ordered list
  $[R_1,R_{B_1},\ldots,R_{B_n}]$. In dimension $1$, this comes down to
  reducing $[x_1^{D_1+1},0]=[B_1,0]=R_{B_1}$ \wrt
  $[R_1,R_{B_1}]$, and thus $R_1$ only. This is indeed the first step of
  the extended Euclidean
  algorithm called on $B_1$ and $F_1$ as described in Section~\ref{ss:bm}.
\end{remark}

\begin{algorithm2e}[htbp!]
  \small
  \DontPrintSemicolon
  \TitleOfAlgo{\divalgo\label{algo:divalgo}}
  \KwIn{A table $\bseq=(\seq_{\bi})_{\bi\in\N^n}$ with coefficients in
    $\K$, a monomial ordering $\prec$ and two monomials $a$ and $b$.}
  \KwOut{A \gb $G$ of the ideal of relations of $\bseq$ for $\prec$.}
  $T \coloneqq \{t\in\cT,\ t\preceq a\}$,
  $U \coloneqq \{u\in\cT,\ u\preceq b\}$.\;
  $M \coloneqq \lcm(T)\,\lcm(U)$.\;
  \lFor(\tcp*[f]{pairs on the edge}){$i$ \KwFrom $1$ \KwTo
    $n$}{$R_{B_i}\coloneqq\Big[x_i^{1+\deg_{x_i} M},0\Big]$.}
  $P \coloneqq \sum\limits_{\tau\in(T+U)}[\tau]\,\frac{M}{\tau}$.\tcp*{the
    mirror of the truncated generating series}
  $R \coloneqq \{[P,1]\}$.\tcp*{set of pairs $[F_m,C_m]
    =[P\cdot C_m\bmod B,C_m]$ to be tested}
  $R'\coloneqq \emptyset$.\tcp*{set of failing pairs}
  $G \coloneqq \emptyset$, $S \coloneqq \emptyset$.
  \tcp*{the future \gb and staircase}
  \While{$R\neq\emptyset$}{
    $R_m=[F_m,C_m]\coloneqq$ first element of $R$ and remove it from $R$.\;
    \uIf(\tcp*[f]{good relation, see
      Proposition~\ref{prop:lt}}){$m\notin T$ or $\LM(\tilde{F}_m)\prec\frac{M}{b\,s}$}{
      $G\coloneqq G\cup\{C_m\}$.\;
    }
    \Else(\tcp*[f]{bad relation}){
      $R'\coloneqq R'\cup\{R_m\}$.\;
      \Forall(\tcp*[f]{reduce next pairs with it}){$r\in R$}{
        $r\coloneqq\NormalForm (r,[R_{B_1},\ldots,R_{B_r},R_m])$.\;
      }
      $S\coloneqq\Stabilize\left(S\cup\left\{m\right\}\right)$.\;
      $H\coloneqq\Border(S)$.\;
      \Forall(\tcp*[f]{compute new pairs}){$h\in H$}{
        \If{there is no relation $C_h\in G$ or no pair $R_h\in R$}{
          Make a new pair $R_h=[F_h,C_h]$ following
          Instruction~\ref{ins:newpair} and add it to $R$.
        }
      }
    }
  }
  \KwRet $G$.
\end{algorithm2e}
\begin{example}[See~\cite{web}]
  Let $\bin=\left(\binom{i}{j}\right)_{(i,j)\in\N^2}$ be the binomial
  table, $\prec$ be the
  $\DRL(y\prec x)$ monomial ordering and $T=\cT[x^3]$ and $U=\{1\}$ be
  sets of terms. We have $a=x^3$, $b=1$,
  $T=T+U$ and $M=\lcm(T)=x^3\,y^3$ so that
  $P_T=x^3\,y^3+x^2\,y^3+x^2\,y^2+x\,y^3+2\,x\,y^2+y^3$, $R_1=[P_T,1]=[F_1,C_1]$
  and $R_{B_1}=[x^4,0]$, $R_{B_2}=[y^4,0]$.

  \begin{itemize}
  \item $m=1$, thus $s=x^3$ and as $\LM(F_1)=\LM(P_T)=x^3\,y^3=\frac{M}{1}$, then the
    relation $C_1$ fails when shifting by $1$ so that $1$ is in the
    staircase. Thus $\cH=\{y,x\}$. We create $R_y$ by computing the normal form of
    $[y\,\LM(F_1),0]=[x^3\,y^4,0]$ \wrt $[R_1,R_{B_1},R_{B_2}]$ and get
    $R_y=[F_y,C_y]=[x^2\,y^3+2\,x\,y^3,y]$. Likewise
    $R_x=[F_x,C_x]=[x^3\,y^2+x^2\,y^2-2\,x\,y^2-y^3,x-1]$.
  \item $m=y$, thus $s=x^2$ and as $\LM(F_y)=x^2\,y^3=\frac{M}{x}$,
    then the relation $C_y$ fails when shifting by $x$ so that $y$ is
    in the staircase. Thus
    $\cH=\{x,y^2\}$. We create
    \begin{itemize}
    \item $R_{y^2}=[0,y^2]$ by computing the normal form of
      $[y\,\LM(F_y),\allowbreak0]=[x^2\,y^4,0]$ \wrt
      $[R_y,R_{B_1},\allowbreak R_{B_2},R_1,R_x]$.
    \end{itemize}
  \item $m=x$, thus $s=x^2$ and as $\LM(F_x)=x^3\,y^2=\frac{M}{y}$,
    then the relation $C_x$ fails when shifting by $y$ so that $x$ is
    in the staircase. Thus
    $\cH=\{y^2,x\,y,x^2\}$. We create
    \begin{itemize}
    \item $R_{x\,y}=[-x^2\,y^2-3\,x\,y^3-2\,x\,y^2-y^3,x\,y-y-1]$ 
      by computing the normal form of
      $R_1$ \wrt $[R_y,R_{B_1},\allowbreak R_{B_2},R_x]$;
    \item $R_{x^2}=[-3\,x^2\,y^2\allowbreak
      -x\,y^3+2\,x\,y^2+y^3,x^2-2\,x+1]$
      by computing the normal form of
      $[x\,\LM(F_x),0]=[x^4\,y^2,0]$ \wrt $[R_x,R_{B_1},R_{B_2},R_1,R_y]$.
    \end{itemize}
  \item $m=y^2$, thus $s=x$ and as $F_{y^2}=0$, then the relation is
    necessarily valid.
  \item $m=x\,y$, thus $s=x$ and as
    $\LM(F_{x\,y})=x^2\,y^2=\frac{M}{x\,y}$, then the relation
    $C_{x\,y}$ is valid.
  \item $m=x^2$, thus $s=x$ and, likewise, as
    $\LM(F_{x^2})=x^2\,y^2=\frac{M}{x\,y}$, then the 
    relation $C_{x^2}$ is valid.
  \end{itemize}
  We return $C_{y^2}=y^2$, $C_{x\,y}=x\,y-y-1$ and $C_{x^2}=x^2-2\,x+1$.
\end{example}
\begin{example}[See~\cite{web} and Example~\ref{ex:sFGLM_bin}]
  We keep $\bin=\left(\binom{i}{j}\right)_{(i,j)\in\N^2}$, the binomial
  table, and $\prec$ set as $\DRL(y\prec x)$. We let however $T=U=\cT[x\,y^2]$
  so that $a=b=x\,y^2$, $2\,T=T+T$ and $M=\lcm(2\,T)=x^4\,y^6$. Thus,
  $P_{2\,T}=x^4\,y^6+x^3\,y^6+x^3\,y^5+x^2\,y^6+2\,x^2\,y^5+x\,y^6+x^2\,y^4
  +3\,x\,y^5+y^6+3\,x\,y^4$, $R_1=[P_{2\,T},1]=[F_1,C_1]$
  and $R_{B_1}=[x^5,0]$, $R_{B_2}=[y^7,0]$.

  \begin{itemize}
  \item $m=1$, thus $s=x\,y^2$ and
    $\cT[b]+\cT[s]=\cT[x^2\,y^4]\setminus\{x^4\,y,x^5\}$, with $x^5$
    not dividing $M$. Therefore,
    $\tilde{F}_1$ is obtained from $F_1$ by removing monomial
    $\frac{M}{x^4\,y}=y^5$. As
    $\LM(F_1)=\LM(P_{2\,T})=x^4\,y^6\succeq\frac{M}{b\,s}=x^2\,y^2$, then the 
    relation $C_1$ fails and $1$ is in the
    staircase. Thus $\cH=\{y,x\}$. We create $R_y$ by computing the normal form of
    $[y\,\LM(F_1),0]=[x^3\,y^4,0]$ \wrt $[R_1,R_{B_1},R_{B_2}]$ and get
    $R_y=[F_y,C_y]=[x^3\,y^6+2\,x^2\,y^6+x^2\,y^5+3\,x\,y^6+3\,x\,y^5,y]$. Likewise
    $R_x=[F_x,C_x]=[x^4\,y^5+x^3\,y^5+x^3\,y^4+x^2\,y^5+2\,x^2\,y^4+3\,x\,y^5
    +y^6+3\,x\,y^4,x-1]$.
  \item $m=y$, thus $s=x\,y$ and
    $\cT[b]+\cT[s]=\cT[x\,y^4]\setminus\{x^4\}$. Therefore,
    $\tilde{F}_y$ is obtained from $F_y$ by removing monomial
    $\frac{M}{x^4}=y^6$. As $\LM(F_y)=x^3\,y^6\succeq\frac{M}{b\,s}=x^2\,y^3$,
    then the relation $C_y$ fails and $y$ is
    in the staircase. Thus
    $\cH=\{x,y^2\}$. We create
    \begin{itemize}
    \item $R_{y^2}=[x^2\,y^6+3\,x\,y^6,y^2]$ by computing the normal form of
      $[y\,\LM(F_y),\allowbreak0]=[x^3\,y^7,0]$ \wrt
      $[R_y,R_{B_1},R_{B_2},R_1,R_x]$.
    \end{itemize}
  \item $m=x$, thus $s=y^2$ and
    $\cT[b]+\cT[s]=\cT[x\,y^4]\setminus\{x^3\,y,x^4\}$. Therefore,
    $\tilde{F}_x$ is obtained from $F_x$ by removing monomials
    $\frac{M}{x^3\,y}=x\,y^5$ and $\frac{M}{x^4}=y^6$.
    As $\LM(F_x)=x^4\,y^5\succeq\frac{M}{b\,s}=x^3\,y^2$,
    then the relation $C_x$ fails and $x$ is
    in the staircase. Thus
    $\cH=\{y^2,x\,y,x^2\}$. We create
    \begin{itemize}
    \item $R_{x\,y}=[-4\,x\,y^6-x^2\,y^4-6\,x\,y^5-y^6-3\,x\,y^4,x\,y-y-1]$ 
      by computing the normal form of
      $R_1$ \wrt $[R_y,R_{B_1},\allowbreak R_{B_2},R_x]$;
    \item $R_{x^2}=[x^4\,y^4\allowbreak
      +x^3\,y^4-4\,x^2\,y^5-x\,y^6-5\,x^2\,y^4+3\,x\,y^5+y^6+3\,x\,y^4,x^2-2\,x+1]$
      by computing the normal form of
      $[x\,\LM(F_x),0]=[x^5\,y^5,0]$ \wrt $[R_x,R_{B_1},R_{B_2},R_1,R_y]$.
    \end{itemize}
  \item $m=y^2$, thus $s=x$ and
    $\cT[b]+\cT[s]=\cT[x^2\,y^2]$.
    As $\LM(F_{y^2})=x^4\,y^6\succeq\frac{M}{b\,s}=x^2\,y^4$,
    then the relation $C_{y^2}$ fails and $y^2$ is
    in the staircase. Thus
    $\cH=\{x\,y,x^2,y^3\}$. We create
    \begin{itemize}
    \item $R_{y^3}=[0,y^3]$ by computing the normal form of
      $[y\,\LM(F_{y^2}),\allowbreak0]=[x^2\,y^7,0]$ \wrt
      $[R_{y^2},\allowbreak R_{B_1},\allowbreak R_{B_2},R_1,R_x,R_y,R_{x^2}]$.
    \end{itemize}
  \item $m=x\,y$, thus $s=y$ and 
    $\cT[b]+\cT[s]=\cT[x\,y^3]\setminus\{x^3\}$. Therefore,
    $\tilde{F}_{x\,y}$ is obtained from $F_{x\,y}$ by removing monomial
    $\frac{M}{x^3}=x\,y^6$.
    As $\LM(\tilde{F}_{x\,y})=x^2\,y^4\prec\frac{M}{b\,s}=x^3\,y^3$,
    $C_{x\,y}$ is valid.
  \item $m=x^2$, thus $s=1$ and $\cT[b]+\cT[s]=\cT[x\,y^2]$.
    As $\LM(F_{x^2})=x^4\,y^4\succeq\frac{M}{b\,s}=x^3\,y^4$,
    then the relation $C_{x^2}$ fails and $x^2$ is
    in the staircase. Thus
    $\cH=\{x\,y,y^3,x^3\}$. We create
    \begin{itemize}
    \item  $R_{x^3}=[-4\,x^3\,y^5-6\,x^3\,y^4+7\,x^2\,y^5+5\,x\,y^6
      +8\,x^2\,y^4-3\,x\,y^5-y^6-3\,x\,y^4,x^3-3\,x^2+y^2+3\,x-1]$ by
      computing the normal form of 
      $[x\,\LM(F_{x^2}),\allowbreak0]=[x^5\,y^4,0]$ \wrt
      $[R_{x^2},\allowbreak R_{B_1},\allowbreak R_{B_2},R_1,R_x,R_y,R_{y^2}]$.
    \end{itemize}
  \item $m=y^3$, thus $s=1$ and 
    $\cT[b]+\cT[s]=\cT[x\,y^2]$.
    As $\LM(\tilde{F}_{y^3})=0\prec\frac{M}{b\,s}=x^3\,y^4$,
    $C_{y^3}$ is valid.
  \item $m=x^3$, thus $s=0$ and 
    $C_{x^3}$ is trivially valid.

    Notice that any relation in $x^3$ would trivially be valid. Though,
    this is the one yielding the smallest leading monomial for
    $F_{x^3}$, \ie $x^3\,y^5$.
  \end{itemize}
  We return $C_{x\,y}=x\,y-y-1$, $C_{y^3}=y^2$ and
  $C_{x^3}=x^3-3\,x^2+y^2+3\,x-1$.
\end{example}

\begin{remark}
  Like the \BMS algorithm, this algorithm creates new potential
  relations by making polynomial combinations of failing relations. As
  a consequence, at each step of the main loop, the potential
  relations, \ie elements of $R$, are not necessarily
  interreduced. Either we can interreduce the final \gb before
  returning it at the last line of the algorithm, or when $C_g$ is
  added to the set $G$ we can update all the current relations by
  removing multiples of $[F_g,C_g]$ and likewise, reduce by $[F_g,C_g]$,
  any subsequent pair $[F_m,C_m]$.
\end{remark}
\begin{example}[See~\cite{web}]
  We give the trace of the \divalgo algorithm with the slight
  modification above called on the table
  $\bseq=((2\,i+1)+(2\,j-1)\,(-1)^{i+j})_{(i,j)\in\N^2}$, the
  stopping monomials $a=y^5$ and $b=1$ and the monomial ordering
  $\DRL(y\prec x)$.

  We set $T\coloneqq\cT[y^5],U\coloneqq\{1\}$, $M\coloneqq x^4\,y^5$ and
  $P=4\,x^3\,y^5+4\,x^4\,y^3+4\,x^3\,y^4+4\,x^2\,y^5-4\,x^4\,y^2+4\,x^2\,y^4
  +8\,x\,y^5+8\,x^4\,y+8\,x^3\,y^2+8\,x^2\,y^3
  +8\,x\,y^4+8\,y^5-8\,x^4$,
  $R_{B_1}\coloneqq[x^5,0],R_{B_2}\coloneqq[y^6,0]$,
  $R=[[P,1]]$.

  \begin{description}
  \item[Pair] $R_1=[F_1,C_1]=[P,1]$, $R\coloneqq\emptyset$ and since $1\in T$ but
    $\LM(F_1)=x^3\,y^5\succeq \frac{M}{s}=x^4$, then
    \begin{itemize}
    \item $R'\coloneqq\{R_1\}$, $S\coloneqq\{1,x\}$ and
      $H\coloneqq\{y,x^2\}$.
    \item We make new pairs added to $R$:
      \begin{itemize}
      \item $R_y=[F_y,C_y]\coloneqq\NormalForm
        ([y\,\LM(F_1),0],[R_1,R_{B_1},\allowbreak R_{B_2}])$
        which can be normalized into
        $R_y,=[4\,x^4\,y^4\allowbreak -\cdots,y-1]$;
      \item $R_{x^2}=[F_{x^2},C_{x^2}]\coloneqq\NormalForm
        ([x^2\,\LM(F_1),0],[R_1,\allowbreak R_{B_1},R_{B_2}])$
        which can be normalized into
        $R_{x^2}=[4\,x^4\,y^3\allowbreak -\cdots,x^2-x-1]$.
      \end{itemize}
    \end{itemize}
  \item[Pair] $R_y=[F_y,C_y]$, $R\coloneqq\{R_{x^2}\}$ and
    since $y\in T$ but
    $\LM(F_y)=x^4\,y^4\succeq\frac{M}{s}=x^4\,y$, then
    \begin{itemize}
    \item $R'\coloneqq\{R_1,R_y\}$, $S\coloneqq\{1,y,x\}$ and
      $H\coloneqq\{y^2,x\,y,x^2\}$.
    \item We make new pairs added to $R$:
      \begin{itemize}
      \item As $y\,\LM(F_y)=x^4\,y^5\not\in\langle x^5,y^6\rangle$ and
        $\LM(F_1)\neq y\,\LM(F_y)$, we can only set 
        $R_{y^2}=[F_{y^2},C_{y^2}]\coloneqq\NormalForm
        (y\,R_y,\allowbreak [R_{B_1},R_{B_2},R_1,R_y])$
        which can be normalized into
        $R_{y^2}=[-4\,x^4\,y^3-\cdots,y^2-x+2\,y-1]$;
      \item $R_{x\,y}=[F_{x\,y},C_{x\,y}]\coloneqq\NormalForm
        ([x\,\LM(F_y),0],[R_y,\allowbreak R_{B_1},R_{B_2},R_1])$
        which can be normalized into
        $R_{x\,y}=[4\,x^4\,y^2-\cdots,x\,y-x+y-1]$.
      \item Nothing is done for $x^2$ since $R_{x^2}$
        already exists.
      \end{itemize}
    \end{itemize}
  \item[Pair] $R_{y^2}=[F_{y^2},C_{y^2}]$,
    $R\coloneqq\{R_{x\,y},R_{x^2}\}$ and since $y^2\in T$ but
    $\LM(F_{y^2})=x^4\,y^3\succeq\frac{M}{s}=x^4\,y^2$, then
    \begin{itemize}
    \item As $\LM(F_{x^2})\succeq\LM(F_{y^2})$, we reduce it and obtain
      $R_{x^2}\coloneqq[-8\,x^2\,y^4\allowbreak -\cdots,x^2+y^2-2\,x+2\,y-2]$.
    \item $R'\coloneqq\{R_1,R_y,R_{y^2}\}$,
      $S\coloneqq\{1,y,x,y^2\}$ and 
      $H\coloneqq\{x\,y,x^2,y^3\}$.
    \item We make new pairs added to $R$:
      \begin{itemize}
      \item $R_{x\,y}$ and $R_{x^2}$ already exist so we do nothing
        for them.
      \item Since $\LM(F_y)=y\,\LM(F_{y^2})$, we can set
        $R_{y^3}=[F_{y^3},C_{y^3}]\allowbreak\coloneqq\NormalForm 
        (R_y,[R_{y^2},R_{B_1},R_{B_2},R_y,R_1])$
        which can be normalized into
        $R_{y^3}=[4\,x^3\,y^4-\cdots,y^3-x\,y+y^2+x-2\,y]$.
      \end{itemize}
    \end{itemize}
  \item[Pair] $R_{x\,y}=[F_{x\,y},C_{x\,y}]$,
    $R\coloneqq\{R_{x^2},R_{y^3}\}$ and since $x\,y\in T$ and
    $\LM(F_{x\,y})=x^4\,y^2\prec\frac{M}{s}=x^2\,y^5$, then
    \begin{itemize}
    \item $G\coloneqq\{x\,y-x+y-1\}$.
    \item As $C_{y^3}=y^3-x\,y+y^2+x-2\,y$ has a term in $x\,y$, we
      update
      $R_{y^3}\coloneqq R_{y^3}+R_{x\,y}=[4\,x^3\,y^4-\cdots,y^3+y^2-y-1]$. 
    \end{itemize}
  \item[Pair] $R_{x^2}=[F_{x^2},C_{x^2}]$,
    $R\coloneqq\{R_{y^3}\}$ and since $x^2\in T$ and
    $\LM(F_{x^2})=x^2\,y^4\prec\frac{M}{s}=x^2\,y^5$, then
    \begin{itemize}
    \item $G\coloneqq\{x\,y-x+y-1,x^2+y^2-2\,x+2\,y-2\}$.
    \end{itemize}
  \item[Pair] $R_{y^3}=[F_{y^3},C_{y^3}]$,
    $R\coloneqq\emptyset$ and since $y^3\in T$ and
    $\LM(F_{y^3})=x^3\,y^4\prec\frac{M}{s}=x^4\,y^3$, then
    \begin{itemize}
    \item $G\coloneqq\{x\,y-x+y-1,x^2+y^2-2\,x+2\,y-2,y^3+y^2-y-1\}$.
    \end{itemize}
  \end{description}
  We return $G$.
\end{example}

\begin{theorem}\label{th:main}
  Let a table $\bseq$, a monomial ordering $\prec$ and two monomials
  $a$ and $b$ be the input of the \divalgo algorithm. Let us
  assume that the reduced \gb $\cG$ of the ideal of relations of $\bseq$ for
  $\prec$ and its
  staircase $S$ satisfy
  $a\succeq\max(S\cup\LM(\cG))$ and for all $g\preceq a$,
  $s=\max\{\sigma\in\cT,\
  \sigma\,g\preceq a\}$, we have
  $\max(S)\preceq s$.
  
  Then, the
  \divalgo algorithm terminates and computes a minimal
  \gb of the ideal of relations of 
  $\bseq$ for $\prec$ in
  $O\big(\Card{S}\,(\Card{S}+\Card{\cG})\,\Card{(\cT[a]+\cT[b])}\big)$
  operations in the base field.
\end{theorem}
\begin{proof}
  The proof is mainly based on the termination and validity of the
  \BMS algorithm. For any monomial $m\in\cT[a]$, we denote by
  $C_m^{\star}$ the last (and therefore one with the largest fail) relation made
  by the \BMS algorithm starting with $m$, if there is any.
  
  Starting with $R_1=[F_1,C_1]=[P_{\cT[a]+\cT[b]},1]$, $\LM(F_1)$
  yields exactly the fail of relation
  $C_1=C_1^{\star}$ so that, as
  in the \BMS algorithm, we know the leading monomials of the potential
  next  relations.

  Let us assume now that for any monomial $\mu\prec h$, the pair
  $R_{\mu}=[F_{\mu},C_{\mu}]$ made by the \divalgo algorithm is equivalent to
  $C_{\mu}^{\star}$, that is either both $C_{\mu}$ and $C_{\mu}^{\star}$ fail
  when shifting by exactly the same monomial or they both succeed on
  $\cT[a]+\cT[b]$.

  Since $C_{\mu}$ and $C_{\mu}^{\star}$ are equivalent, the
  current discovered staircase by the \BMS and the \divalgo
  algorithms are
  the same. Thus either $h$ is a leading monomial of a relation to be
  built by both algorithms or it is not. Without loss of generality, we
  can assume it is. There exists a monomial $m$ such that $m|h$ and
  $R_m=[F_m,C_m]$ and $C_m^{\star}$ have been made.
  In the \BMS
  algorithm, the relation $C_h^{\star}$ is obtained as
  $\frac{h}{m}\,C_m^{\star}-\sum_{\mu\prec h}
  q_{\mu}^{\star}\,C_{\mu}^{\star}$
  while in the \divalgo algorithm, $C_h$ is made as $\frac{h}{m}\,C_m-\sum_{\mu\prec h}
  q_{\mu}\,C_{\mu}$.
  In each computation, $q_{\mu}^{\star}$ and
  $q_{\mu}$ are chosen so that $C_m^{\star}$ and $C_m$ have the
  largest fail (or equivalently $\tilde{F}_m$ has the least leading monomial),
  hence $C_m^{\star}$ and $C_m$ are equivalent.
  For $h\in S$, the potential relation $C_h$ made by the algorithm
  must fail when shifted by a monomial in $S$. Thus,
  there exist $\sigma_1,\sigma_2$ such that $\sigma_1\,\sigma_2\in S$,
  $\sigma_1\,h\preceq a$, $\sigma_2\preceq b$ and the
  column labeled with $\sigma_1\,h$ of the matrix $H_{\cT[b],\cT[a]}$ is
  independent from the previous ones.
  For $g\in\LM(\cG)$, by Section~\ref{ss:multidim}, the relation $C_g$
  has been tested shifted by all the monomials in $\cT[b]+\cT[s]$, with
  $s=\max\{\sigma\in\cT,\ \sigma\,g\preceq a\}$. The theorem hypothesis
  is exactly that the full staircase is included in the set of
  tested shifts, hence we can ensure that $C_g$ corresponds to
  a kernel vector of $H_{S,S\cup\{g\}}$ with the last coordinate equal
  to $1$.

  Furthermore, as in the proof of the \BMS algorithm, the failures of
  relations $C_m$ ensure that the returned relations $C_h$ are all
  such that their leading monomials are minimal for the division. That
  is, the returned \gb is minimal.

  Concerning the complexity of the algorithm. Since $\cT[a]$ and
  $\cT[b]$ are stable by division, so is $\cT[a]+\cT[b]$. Let us
  recall that the support of $P_{\cT[a]+\cT[b]}$ is
  $\big\{\frac{M}{\tau},\ \tau\in(\cT[a]+\cT[b])\big\}$, $M=\lcm(\cT[a]+\cT[b])$.
  Since
  each $F_m$ satisfies $F_m=P_{\cT[a]+\cT[b]}\,C_m\bmod B$, then the
  monomials in the support of $F_m$ are multiples of the monomials in
  the support of $P_{\cT[a]+\cT[b]}$ and thus are included in the
  support of $P_{\cT[a]+\cT[b]}$. Each pair $R_m=[F_m,C_m]$ for
  $m\in S\cup\LM(\cG)$ must be
  reduced by all the previous ones lying in the staircase in at most
  $\Card{S}\,\Card{(\cT[a]+\cT[b])}$ operations. Reducing the
  relations to obtain a minimal \gb can be done in
  $O(\Card{S}\,\Card{\cG}\,\Card{(\cT[a]+\cT[b])})$ operations, hence this part
  is not the bottleneck of the algorithm.
\end{proof}

\begin{remark}
  Using the same notation, the \AGbb algorithm computes a border
  basis $\cB$ of the ideal of relations using
  $O\pare{\Card{S}\pare{\Card{S}+\Card{\cB}}\,\Card\pare{\cT[a]+\cT[a]}}$
  operations in the base field~\cite{Mourrain:2017:FAB:3087604.3087632}.
  Thus,
  in theory, the \AGbb and the \divalgo algorithms share the same
  complexity estimates, whenever $a=b$. 

  However, the given complexity bound is based on naive
  multivariate polynomial arithmetic. Thus the goal would be to
  investigate how to exploit fast polynomial multiplication to speed up
  the $\NormalForm$ procedure computations in the \divalgo
  algorithm.

  Let us recall that in the univariate case, complexity improvements
  are made thanks to fast Euclidean
  algorithm through a divide-and-conquer approach and using fast
  polynomial division and multiplication.

  In this multivariate setting, a divide-and-conquer
  approach has already been investigated in~\cite{NaldiN2020} relying on
  fast polynomial matrix arithmetic. Likewise, some improvements were made  
  regarding the reduction of a bivariate polynomial by the reduced \gb of the
  ideal spanned by two polynomials for $\DRL$ in~\cite{vdH:ggg}. This
  is a first step in this direction.

  A further step would be to determine the quotients fastly using, like
  in the univariate case, fast multiplication. Usually the reduction
  of a polynomial by
  several polynomials might be intricate. Yet, in our experiments, we
  observed that
  the call
  $R_h=\NormalForm(R_{m'},[R_m,R_{B_1},\ldots,R_{B_n},R_{t_1},\ldots,\allowbreak
  R_{t_r}])$
  can actually be done in several simpler steps.
  \begin{enumerate}
  \item A call to $R_h=\NormalForm(R_{m'},R_m)$ to reduce $F_m$
    \wrt $F_{m'}$.
  \item A cleaning step to remove some high-degree
    monomials in $F_h$, corresponding to
    $R_h=\NormalForm(R_h,[R_{B_1},\ldots,R_{B_n}])$. Note that the
    quotients, here, need not be stored.
  \item  A Gaussian elimination-like step to find the linear combination
    of $F_h,F_{t_1},\ldots,F_{t_r}$ with
    the smallest leading monomial. This corresponds to a call
    $R_h=\NormalForm(R_h,[R_{t_1},\ldots,R_{t_r}])$.
  \end{enumerate}
  Hence, one might
  only need to compute the first quotient, associated to $F_m$, fastly.
  
  Finally, the computation of $C_h$ is done through polynomial
  multiplications and thus benefit from any improvement thereof.
  
  We shall see in
  Section~\ref{s:bench}, that
  the \divalgo algorithm seems to perform better thanks to the
  multivariate polynomial arithmetic.
\end{remark}


%% file: 5-adaptive.tex
In some applications, the actual size of the staircase, or at least an
upper bound thereof, is known. While it provides an early termination
criterion for the \BMS, \sFGLM and \divalgo algorithms, this might
fail to drastically reduce the number of table queries.
Indeed, for the $\DRL(x_n\prec\cdots\prec x_1)$ ordering,
whether the set of leading monomials of the \gb is $\{x_n,\ldots,x_2,x_1^d\}$
or all the monomials of degree $d$:
$\{x_n^d,x_{n-1}\,x_n^{d-1},\ldots,\allowbreak x_1\,x_n^{d-1},\ldots,x_1^d\}$,
the \BMS algorithm requires to visit all the monomials up to
$x_1^{2\,d-1}$.  Therefore, it needs to visit
$\binom{n+2\,d-1}{n}$ table terms to compute a \gb
of size $n$
with a staircase of size $d$ in the former case and a \gb of size
$\binom{n+d-1}{n-1}$ with a staircase
of size $\binom{n+d-1}{n}$ in the latter. Furthermore, in some applications, like
the \spFGLM algorithm one, computing a single table term can be very
costly. Thus, requiring as few table terms as possible to retrieve the
correct ideal of relations is critical.

The \asFGLM algorithm~\cite{issac2015} was designed to minimize the
number of table queries by taking into
account the shape of the \gb gradually as it is discovered. The
algorithm starts with $S=\emptyset$. At each step, $S$ is a staircase
containing only monomials that we know are in the target staircase,
this means that the matrix $H_{S,S}$ must be full rank.
Likewise, $L$ is a set of monomials on the border of $S$. For $m$ the
smallest monomial in $L$, we check if $H_{S\cup\{m\},S\cup\{m\}}$, with
$S\cup\{m\}$ has a greater rank than $H_{S,S}$ or not. If it does
not, then the last column, labeled with $m$, must be linearly
dependent from the previous one. That is, a relation $C_m$ is found
and any multiple of $m$ is removed from $L$.
Otherwise, no relation $C_m$ with $\LM(C_m)=m$ must exist. Thus, 
$m$ is added to the staircase $S$, removed from $L$ and monomials
$m\,x_i$ are added to $L$.

\begin{example}
  Let us consider the sequence $\bseq=(p_{i_0+1})_{\bi\in\N^n}$ where
  $p_{i_0+1}$ stands for the $(i_0+1)$st prime number if $i_0<d$ and
  $0$ otherwise.
  For $\DRL$, or even $\LEX$, the \asFGLM algorithm
  computes the rank of the following matrices
  \begin{itemize}
  \item $H_{\{1\},\{1\}}=\pare{
      \begin{smallmatrix}
        2
      \end{smallmatrix}}$. Its rank is $1$, the dimension of the matrix, so $1\in S$;
  \item $H_{\{1,x_n\},\{1,x_n\}}=\pare{
      \begin{smallmatrix}
        2	&2\\
        2	&2
      \end{smallmatrix}}$. Its rank is $1$ which is not the dimension
    of the matrix so a relation $C_{x_n}$ is found. This is $C_{x_n}=x_n-1$.
  \item $H_{\{1,x_{n-1}\},\{1,x_{n-1}\}}=\cdots=H_{\{1,x_2\},\{1,x_2\}}=\pare{
      \begin{smallmatrix}
        2	&2\\
        2	&2
      \end{smallmatrix}}$. Their ranks are also $1$ which are not the
    dimensions of the matrices so the relations
    $C_{x_{n-1}}=x_{n-1}-1,\ldots,C_{x_2}=x_2-1$ are found.
  \item $H_{\{1,x_1\},\{1,x_1\}}=\pare{
      \begin{smallmatrix}
        2	&3\\
        3	&5
      \end{smallmatrix}}$. Its rank is $2$, the dimension of the
    matrix. Thus, $x_1\in S$.
  \item $H_{\{1,x_1,x_1^2\},\{1,x_1,x_1^2\}}=\pare{
      \begin{smallmatrix}
        2	&3	&5\\
        3	&5	&7\\
        5	&7	&11
      \end{smallmatrix}}$. Its rank is $3$, the dimension of the
    matrix. Therefore, $x_1^2\in S$.
  \item \ldots;
  \item $H_{\{1,x_1,\ldots,x_1^{d-1}\},\{1,x_1,\ldots,x_1^{d-1}\}}=\pare{
      \begin{smallmatrix}
        2	&3	&5	&\cdots	&p_d\\
        3	&5	&7	&\cdots	&0\\
        5	&7	&11	&\cdots	&0\\
        \vdots	&\vdots	&\vdots	&\ddots	&\vdots\\
        p_d	&0	&0	&\cdots	&0
      \end{smallmatrix}}$. Its rank is $d$, the dimension of the
    matrix, so $x_1^d\in S$.
  \item $H_{\{1,x_1,\ldots,x_1^d\},\{1,x_1,\ldots,x_1^d\}}=\pare{
      \begin{smallmatrix}
        2	&3	&5	&\cdots	&p_d	&0\\
        3	&5	&7	&\cdots	&0	&0\\
        5	&7	&11	&\cdots	&0	&0\\
        \vdots	&\vdots	&\vdots	&\ddots	&\vdots	&\vdots\\
        p_d	&0	&0	&\cdots	&0	&0\\
        0	&0	&0	&\cdots	&0	&0
      \end{smallmatrix}}$. Its rank is also $d$, which is not the
    dimension of the matrix. Thus, the relation
    $C_{x_1^d}=x_1^d$ is found.
  \end{itemize}
  It thus requires merely $2\,(n+d)-1$ table terms instead of
  $\binom{n+2\,d-1}{n}$.
\end{example}
As described in Sections~\ref{s:Mat2Pol} and~\ref{s:division}, the
\divalgo algorithm is based on polynomials from matrices with columns set $\cT[a]$
and rows set $\cT[b]$. However, here, we need matrices with more
general sets of monomials for the rows and columns. Therefore, the
main tool of the adaptive variant of the \divalgo algorithm is new
basic routines so that we
can perform polynomial divisions while also ensuring that our
polynomials are related to
multi-Hankel matrices with these columns and rows sets of
monomials. Since at each step, $S$ is a staircase and
$m$ is a monomial lying on its border so that $S\cup\{m\}$ is also a
staircase, then the corresponding matrices would be $H_{S,S}$ and
$H_{S\cup\{m\},S\cup\{m\}}$. Therefore, any instance of $T+U$ from the
previous sections will be replaced by $2\,S\coloneqq S+S$, the
Minkowski sum of $S$ with itself, or likewise by $2\,(S\cup\{m\})$.

At each step, we have the polynomial $P_{2\,S}$, associated to matrix
$H_{S,S}$, a monomial ideal
$B_{2\,S}$, determined as in Section~\ref{ss:multidim}, and pairs
$R_{2\,S,t}=[F_{2\,S,t},C_t]=[P_{2\,S}\,C_t\bmod B_{2\,S},C_t]$. At the next step, we
compute $P_{2\,(S\cup\{m\})}$ from $P_{2\,S}$ by shifting it by
$\frac{\LCM(S\cup\{m\})^2}{\LCM(S)^2}$
and adding the missing terms, update $B_{2\,S}$ into
$B_{2\,(S\cup\{m\})}$ and likewise update each
$R_{2\,(S\cup\{m\}),t}$ by shifting $F_{2\,S,t}$ and adding the
missing terms to make $F_{2\,(S\cup\{m\}),t}$. It remains to compute
$R_{2\,(S\cup\{m\}),m}$ such that $\supp C_m\subseteq S\cup\{m\}$ and
$F_{2\,(S\cup\{m\})}$ is sufficiently ``small'' to ensure that $C_m$ is a
valid relation or not. This is where two approaches are possible. The
first one is the
naive one; it was proposed in~\cite{issac2018} and is recalled in
Section~\ref{ss:naive_adaptive}. It is based on polynomial arithmetic, yet does not
use as many polynomial divisions as the \divalgo does.
The other one is presented in
Section~\ref{ss:adivalgo} and fully uses polynomial divisions to
perform the computations. In particular, we design some new
basic routines to do so. This yields
the \Adivalgo algorithm, or Algorithm~\ref{algo:Adivalgo}.

\subsection{A naive approach}
\label{ss:naive_adaptive}
In this naive approach, the only remaining things are the initialization
and reduction of $R_{2\,(S\cup\{m\}),m}$ and then testing if $C_m$ is valid.

The pair $R_{2\,(S\cup\{m\}),m}$ is 
initialized as $m\,R_{2\,(S\cup\{m\}),1}$ and then reduced, as in
Section~\ref{s:division}, by
$R_{2\,(S\cup\{m\}),B_1},\allowbreak\ldots,R_{2\,(S\cup\{m\}),B_n}$, where
$B_1,\ldots,\allowbreak B_n\in B_{2\,S}$. Then, we can reduce it
incrementally by all the other $R_{2\,(S\cup\{m\}),t}$, which comes
down to just subtracting a constant multiple of them.

Finally, when a valid relation $C_g$ is found, any multiple of $g$ is removed from
the set of potential monomials to add to $S$. Moreover, we
can further reduce a future relation
$R_{2\,(S\cup\{m\}),m}=[F_{2\,(S\cup\{m\}),m},C_m]$ with any pair 
$\Big[\frac{M}{\mu\,g},0\Big]$, $m\succeq\mu\,g$ to clean the support
of $F_{2\,(S\cup\{m\})}$. As $\frac{M}{\mu\,g}$ is a single monomial,
this can be done easily, not unlike the reductions by
$R_{2\,(S\cup\{m\}),B_1},\ldots,\allowbreak R_{2\,(S\cup\{m\}),B_n}$.

\subsection{A division-based adaptive variant}\label{ss:adivalgo}
In~\cite{issac2018}, we did not go further in
the design of an adaptive variant of the \divalgo. In particular, that
version could not initialize a pair $R_{2\,(S\cup\{m\}),m}$ as the quotient of two
pairs of polynomials. In the following part of this section, we show
how to initialize a new pair as the quotient of two previously
computed pairs of polynomials.

The algorithm uses some new procedures that are not needed in
\divalgo. Indeed, at step $m$, when computing
$R_{2\,S',m}=[F_{2\,S',m},C_m]$, with $S'=S\cup\{m\}$, we can at first only ensure that
$\supp C_m\subseteq\cT[m]$. Yet, we need to have $\supp C_m\subseteq
S'=\cT[m]\setminus\langle\LM(C_g),g\in G\rangle$. Thus, we need to
reduce $C_m$ by all the already computed $C_g$'s.
\begin{description}
\item
  [$\NormalFormRightSide(R_{2\,S,m},{[R_{2\,S,g_1},\ldots,R_{2\,S,g_r}]})$]
  that computes
  the quotients $Q_{g_1},\ldots,Q_{g_r}$ of the division of $C_m$ by
  $C_{g_1},\ldots,C_{g_r}$, with $R_{2\,S,g_1},\ldots,R_{2\,S,g_r}\in G$ and
  then returns $R_{2\,S,m}-Q_{g_1}\,R_{2\,S,g_1}-\cdots-Q_{g_r}\,R_{2\,S,g_r}$.

\item [$\NormalFormHigherPart
  (R_{2\,S,m},{[R_{2\,S,s_1},\ldots,R_{2\,S,s_q}]})$] that behaves like $\NormalForm$, except
  only the higher part $\tilde{F}_{2\,S',t}$ of a polynomial
  $F_{2\,S',t}$ is used. For $t=B_i$,
  $\tilde{F}_{2\,S',t}=F_{2\,S',t}=B_i$, otherwise
  $\tilde{F}_{2\,S',t}$ is obtained from $F_{2\,S',t}$ by removing any
  monomial dividing $\frac{\LCM(S')^2}{g}$ with $g\in\LM(\cG)$. Then
  $\NormalFormHigherPart
  (R_{2\,S',m},[R_{2\,S',s_1},\ldots,R_{2\,S',s_q}])$
  computes the normal form of $\tilde{F}_{2\,S',m}$ with respect to
  $[\tilde{F}_{2\,S',s_1},\ldots,\tilde{F}_{2\,S',s_q}]$ and the
  corresponding quotients $Q_1,\ldots,Q_q$. It then returns
  $R_{2\,S',m}-Q_1\,R_{2\,S',s_1}-\cdots-Q_q\,R_{2\,S',s_q}$.
\end{description}
The definition of $\tilde{F}_{2\,S',m}$ extends the one used in
Section~\ref{sss:sfglm}. The rationale is the same: using the leading
monomial of $F_{2\,S',m}$ to check whether $C_m$ is a valid relation
or not is not the same as checking if the last column of $H_{S',S'}$
linearly depends from the previous ones. Indeed, the leading monomial
might correspond to a row that is not present in $H_{S',S'}$. Since
$S'$ does not contain any monomial that is a multiple of $g$ with
$R_g\in G$, then it makes sense to remove any monomial dividing
$\frac{\LCM(S')^2}{g}$, with $R_g$ still in $G$.

\begin{algorithm2e}[htbp!]
  \small
  \DontPrintSemicolon
  \TitleOfAlgo{\Adivalgo\label{algo:Adivalgo}}
  \KwIn{A table $\bseq=(\seq_{\bi})_{\bi\in\N^n}$ with coefficients in
    $\K$, a monomial ordering $\prec$.}
  \KwOut{A \gb $G$ of the ideal of relations of $\bseq$ for $\prec$.}
  $L \coloneqq \{1\}$.\;
  $G \coloneqq \emptyset$, $S \coloneqq \emptyset$. \tcp*{the future
    \gb and staircase}
  $B \coloneqq \{B_1,\ldots,B_n\}=\{x_1^0,\ldots,x_n^0\}$.\;
  $R_{2\,S,B_1} \coloneqq [B_1,0],\ldots, R_{2\,S,B_n} \coloneqq [B_n,0]$.\;
  \While{$L\neq\emptyset$}{
    $m:= $ first element of $L$ and remove it from $L$.\;
    $S':= S\cup\{m\}$.\;
    Update $B$ and all pairs $R_{2\,S,t}$ into $R_{2\,S',t}$ for $t\in S\cup
    G\cup B$.\;
    \lIf{$m=1$}{
      $R_1 \coloneqq [w_{0,\ldots,0},1]$.
    }
    \uElseIf{$m=\mu\,x_i^2$ with $\mu,\mu\,x_i\in S$}{
      $R_{2\,S',m} \coloneqq \NormalFormHigherPart
      (R_{2\,S',\mu},[R_{2\,S',\mu\,x_i},R_{2\,S',B_1},\ldots,R_{2\,S',B_n}])$.
    }
    \Else(\tcp*[f]{$m=\mu\,x_i$ with $\mu\in S$}){
      $R_{2\,S',m} \coloneqq \NormalForm
      (x_i\,R_{2\,S',\mu},[R_{2\,S',B_1},\ldots,R_{2\,S',B_n}])$.
    }
    $R_{2\,S',m} \coloneqq \NormalFormRightSide
    (R_{2\,S',m},[R_{2\,S',g_1},\ldots,R_{2\,S',g_r}])$.\tcp*{$g_1,\ldots,g_r\in G$}
    $R_{2\,S',m} \coloneqq \NormalForm (R_{2\,S',m},[R_{2\,S',B_1},\ldots,R_{2\,S',B_n}])$.\;
    $R_{2\,S',m} \coloneqq \NormalFormHigherPart
    (R_{2\,S',m},[R_{2\,S',s_1},\ldots,R_{2\,S',s_q}])$.\tcp*{$s_1,\ldots,s_q\in S$}
    $R_{2\,S',m}=[F_{2\,S',m},C_m]$.\;
    \uIf(\tcp*[f]{Relation succeeds}){$\LM(\tilde{F}_{2\,S',m})
      \prec\frac{\LCM(S')^2}{m}$}{
      Update $R_{2\,S',m}$ into $R_{2\,S,m}$.\;
      $G \coloneqq G\cup\{R_{2\,S,m}\}$.\;
      Remove multiples of $m$ in $L$.\;
    }
    \Else(\tcp*[f]{Relation fails}){
      Delete every pair $R_{2\,S,t}$ for $t\in S\cup G\cup B$.\;
      $S \coloneqq S'$.\;
      $L \coloneqq L\cup\{x_1\,m,\ldots,x_n\,m\}$, remove any
      multiples of $\LM(G)$ and sort it by
      increasing order.
    }
  }
  \Return $G$.
\end{algorithm2e}
\begin{remark}
  To simplify the presentation, we did not consider the case where
  some sequence terms might not be available. This can happen for
  instance in the error correcting code application.
  Likewise, if the sequence is not linear recurrent, then an infinite
  loop might happen. Both situations require an easy modification of
  the algorithm.
\end{remark}

For the correctness of the algorithm, we need to prove the following lemma.
\begin{lemma}
  Let $S$ be the current computed staircase and let $\cG$ be the
  current computed \gb. Let $m$ be the least monomial not in $S$ and
  not divisible by $\LM(\cG)$ and let $S'=S\cup\{m\}$. Then, for all $t$,
  $\LM(\tilde{F}_{2\,S',t})=\frac{\LCM(S')^2}{\LCM(S)^2}\,\LM(\tilde{F}_{2\,S,t})$ and
  $\LM(\tilde{F}_{2\,S',m})\preceq\frac{\LCM(S')^2}{m}$.
  Furthermore, this weak inequality is an equality if, and only if, $C_m$ is not a
  valid relation.
\end{lemma}
\begin{proof}
  If $\cG=\emptyset$, then the result is clear as
  $\tilde{F}_{2\,S',m}=F_{2\,S',m}$.
  Otherwise, the leading monomial of $\tilde{F}_{2\,S',m}$ is less than
  $\frac{\LCM(S')}{t}$ with $t\in S$ and cannot be a divisor of
  $\frac{\LCM(S')}{g}$ with $g\in\LM(\cG)$. Hence it must be at most
  $\frac{\LCM(S')}{m}$.

  Furthermore, for $S=\{1,\ldots,s_q\}$, the coefficient of
  $\frac{\LCM(S')}{m}$ can be read as the last coefficient of this product
  \[\kbordermatrix{
    	&1	&\cdots	&s_q	&m\\
  1	&[1]	&\cdots	&[s_q]	&[m]\\
  \vdots&\vdots	&	&\vdots	&\vdots\\
  s_q	&[s_q]	&\cdots	&[s_q^2]&[m\,s_q]\\
  m	&[m]	&\cdots	&[m\,s_q]&[m^2]
  }\,
  \begin{pmatrix}
    \alpha_1\\\vdots\\\alpha_{s_q}\\1
  \end{pmatrix}=
  \begin{pmatrix}
    0\\\vdots\\0\\f_{2\,S',m}
  \end{pmatrix}.\]
  Hence, the relation is valid if, and only if, this coefficient is zero.
\end{proof}

The leading monomial of
the higher part of $F_{2\,S',m}$ tells us whether $C_m$ is a valid
relation or not. 

\begin{theorem}\label{th:Adivalgo}
  Assuming the \Adivalgo algorithm called on table $\bseq$ returns a \gb $\cG$
  with staircase $S$, then the algorithm does not need more that
  $\Card{2\,(S\cup\LM(\cG))}$ table queries to recover $\cG$.

  Furthermore, it performs at most $O((\Card{S}+\Card{\cG})^2\,\Card{2\,S})$
  operations to recover $\cG$.
\end{theorem}
\begin{proof}
  Since $\tilde{F}_{2\,S,t}$ is obtained from $F_{2\,S,t}$ by removing
  all the
  monomials dividing $\frac{\LCM(S)^2}{g}$, $g\in\LM(\cG)$, then
  $\tilde{F}_{2\,S,t}$ is actually the polynomial obtained from the
  product of the extended matrix $H_{T,S}$ and the vector representing
  $C_t$, where $T=2\,S\setminus\{g\,\tau, g\in\LM(\cG),\tau\in\cT\}$. Hence,
  updating from $\tilde{F}_{2\,S,t}$ to $\tilde{F}_{2\,S',t}$ is
  equivalent to updating $H_{T,S}$ to $H_{T',S'}$, with
  $T'=2\,S'\setminus\{g\,\tau, g\in\LM(\cG),\tau\in\cT\}$ in this matrix-vector
  product. Yet, the first nonzero coefficient of this product remains
  the same through this updating process.

  At step $m$, $F_{2\,S',t}=P_{2\,S'}\,C_t\bmod B_{2\,S'}$ for any
  $t\in S\cup\LM(\cG)$. As in the end of the step, $m$ is either found
  to be a member of the final $S$ or the final $\LM(\cG)$, then
  $S'\subseteq(S\cup\LM(\cG))$ and thus $2\,S'\subseteq
  2\,(S\cup\LM(\cG))$.

  At step $m$, a pair $R_{2\,S,t}=[F_{2\,S,t},C_t]$ is updated into
  $R_{2\,(S\cup\{m\}),t}$ by adding terms with support in
  $2\,(S\cup\{m\})$, so in the last step, all the polynomials have
  support in a set of size $\Card{2\,(S\cup\{g\})}$ where $g\in\LM(\cG)$.

  Furthermore, for each monomial $t\in S\cup\LM(\cG)$, the pair
  $R_{2\,S',t}$ is reduced by all the previous ones lying in the
  staircase or the \gb in at most $(\Card{S}+\Card{\cG})\,\Card{2\,S}$
  operations. Hence, at most $O((\Card{S}+\Card{\cG})^2\,\Card{2\,S})$
  operations are needed.
\end{proof}

\begin{example}
  We consider the following sequence $\bseq=(\seq_{i,j})_{(i,j)\in\N^2}$
  \[\bseq=
  \begin{pmatrix}
    6	&9	&5	&1	&10	&-6	&-9	&\cdots\\
    3	&12	&2	&4	&7	&-3	&-12	&\cdots\\
    6	&9	&5	&1	&10	&-6	&-9	&\cdots\\
    3	&12	&2	&4	&7	&-3	&-12	&\cdots\\
    \vdots	&\vdots	&\vdots	&\vdots	&\vdots	&\vdots	&\vdots	&\ddots
  \end{pmatrix}\]
  For a pair $R_{2\,S,t}=[F_{2\,S,t},C_t]$, we let
  $\tilde{R}_{2\,S,t}=[\tilde{F}_{2\,S,t},C_t]$.

  We start with $L=\{1\}$, $G=S=\emptyset$, $B=(1,1)$ and
  $R_{2\,S,B_1}=[1,0]$, $R_{2\,S,B_2}=[1,0]$.
  \begin{enumerate}
  \item 
    We set $m=1$, $S'=\{1\}$, $R_{2\,S',B_1}=[1,0]$,
    $R_{2\,S',B_2}=[y,0]$ and initialize
    $R_{2\,S',1}=[6,1]$.

    As $\tilde{F}_{2\,S',1}=F_{2\,S',1}=1$ and
    $\LM(\tilde{F}_{2\,S',1})=1=\frac{\LCM(S')^2}{1}$ the relation
    fails.

    $L$ is updated to $\{y,x\}$.
  \item We set $m=y$, $S'=\{1,y\}$, $R_{2\,S',B_1}=[x,0]$,
    $R_{2\,S',B_2}=[y^3,0]$ and $R_{2\,S',1}=[6\,y^2+9\,y+5,1]=\tilde{R}_{2\,S',1}$.

    We
    initialize $R_{2\,S',y}=[9\,y^2+5\,y,y]$ using $\NormalForm$ on
    $y\,R_{2\,S',1}$ and then reduce it to
    $[-\frac{17}{2}\,y-\frac{15}{2},y-\frac{3}{2}]$.

    As $\tilde{F}_{2\,S',y}=F_{2\,S',y}$ and
    $\LM(\tilde{F}_{2\,S',y})=\frac{\LCM(S')^2}{y}$ the relation
    fails.

    $L$ is updated to $\{x,y^2\}$.
  \item We set $m=x$, $S'=\{1,y,x\}$, $R_{2\,S',B_1}=[x^3,0]$,
    $R_{2\,S',B_2}=[y^3,0]$,
    \begin{itemize}
    \item
      $R_{2\,S',1}=[6\,x^2\,y^2+9\,x^2\,y+3\,x\,y^2+5\,x^2+12\,x\,y+6\,y^2,1]
      =\tilde{R}_{2\,S',1}$,
    \item $R_{2\,S',y}=[-\frac{17}{2}\,x^2\,y+\frac{15}{2}\,x\,y^2-\frac{15}{2}\,x^2
      -18\,x\,y-9\,y^2,y-\frac{3}{2}]=\tilde{R}_{2\,S',y}$.
    \end{itemize}
    We
    initialize
    $R_{2\,S',x}=[3\,x^2\,y^2+12\,x^2\,y+6\,x\,y^2, x]$ using
    $\NormalForm$ on $x\,R_{2\,S',1}$ 
    and then reduce it to
    $[\frac{189}{17}\,x\,y^2-\frac{155}{17}\,x^2-\frac{372}{17}\,x\,y
    -\frac{186}{17}\,y^2, x+\frac{15}{17}\,y-\frac{31}{17}]$.

    As $\tilde{F}_{2\,S',x}=F_{2\,S',x}$ and
    $\LM(\tilde{F}_{2\,S',x})=\frac{\LCM(S')^2}{x}$ the relation
    fails.
    
    $L$ is updated to $\{y^2,x\,y,x^2\}$.
  \item We set $m=y^2$, $S'=\{1,y,x,y^2\}$, $R_{2\,S',B_1}=[x^3,0]$,
    $R_{2\,S',B_2}=[y^5,0]$,
    \begin{itemize}
    \item $R_{2\,S',1}=[6\,x^2\,y^4+\cdots,1]=\tilde{R}_{2\,S',1}$,
    \item
      $R_{2\,S',y}=[-\frac{17}{2}\,x^2\,y^3+\cdots,y-\frac{3}{2}]=\tilde{R}_{2\,S',y}$,
    \item $R_{2\,S',x}=[\frac{189}{17}\,x\,y^4+\cdots,
      x+\frac{15}{17}\,y-\frac{31}{17}]=\tilde{R}_{2\,S',x}$.
    \end{itemize}
    We
    initialize
    $R_{2\,S',y^2}=[-\frac{106}{17}\,x\,y^4+\cdots,y^2-\frac{13}{17}\,y+\frac{16}{51}]$
    using $\NormalFormHigherPart$ on $R_{2\,S',1}$ 
    and $R_{2\,S',y}$
    and then reduce it to
    $[\frac{1381}{189}\,x^2\,y^2+\cdots, y^2+\frac{106}{189}\,x
    -\frac{17}{63}\,y-\frac{134}{189}]$.

    As $\tilde{F}_{2\,S',y^2}=F_{2\,S',y^2}$ and
    $\LM(\tilde{F}_{2\,S',y^2})=\frac{\LCM(S')^2}{y^2}$ the relation
    fails.
    
    $L$ is updated to $\{x\,y,x^2,y^3\}$.
  \item We set $m=x\,y$, $S'=\{1,y,x,y^2,x\,y\}$,
    $R_{2\,S',B_1}=[x^3,0]$,
    $R_{2\,S',B_2}=[y^5,0]$,
    \begin{itemize}
    \item $R_{2\,S',1}=[6\,x^2\,y^4+\cdots,1]=\tilde{R}_{2\,S',1}$,
    \item
      $R_{2\,S',y}=[-\frac{17}{2}\,x^2\,y^3+\cdots,y-\frac{3}{2}]=\tilde{R}_{2\,S',y}$,
    \item $R_{2\,S',x}=[\frac{189}{17}\,x\,y^4+\cdots,
      x+\frac{15}{17}\,y-\frac{31}{17}]=\tilde{R}_{2\,S',x}$,
    \item $R_{2\,S',y^2}=[\frac{1381}{189}\,x^2\,y^2+\cdots,
      y^2+\frac{106}{189}\,x -\frac{17}{63}\,y-\frac{134}{189}]=\tilde{R}_{2\,S',y^2}$.
    \end{itemize}
    We
    initialize
    $R_{2\,S',x\,y}=[-\frac{15}{2}\,x^2\,y^4+\cdots,x\,y-\frac{3}{2}\,x]$
    using $\NormalForm$ on $x\,R_{2\,S',y}$
    and then reduce it to
    $[-15\,y^4-7\,x^2\,y-x\,y^2-14\,y^3-10\,x^2-4\,x\,y-5\,y^2, x\,y+x-y-1]$.

    As $\tilde{F}_{2\,S',x\,y}=F_{2\,S',x\,y}$ and
    $\LM(\tilde{F}_{2\,S',x\,y})=y^4\prec\frac{\LCM(S')^2}{x\,y}$ the relation
    succeeds!

    We update $R_{2\,S,x\,y}=[-9\,x\,y^4+\cdots,x\,y+x-y-1]$ and put
    it in $G$.

  \item We set $m=x^2$, $S'=\{1,y,x,y^2,x^2\}$,
    $R_{2\,S',B_1}=[x^5,0]$, $R_{2\,S',B_2}=[y^5,0]$

    \begin{itemize}
    \item $R_{2\,S',1}=[6\,x^4\,y^4+\cdots,1]$, $\tilde{R}_{2\,S',1}=[6\,x^4\,y^4+\cdots,1]$,
    \item
      $R_{2\,S',y}=[-\frac{17}{2}\,x^4\,y^3-+\cdots,y-\frac{3}{2}]$,
      $\tilde{R}_{2\,S',y}=[-\frac{17}{2}\,x^4\,y^3-+\cdots,y-\frac{3}{2}]$,
    \item $R_{2\,S',x}=[\frac{189}{17}\,x^3\,y^4+\cdots,
      x+\frac{15}{17}\,y-\frac{31}{17}]$,
      $\tilde{R}_{2\,S',x}=[\frac{189}{17}\,x^3\,y^4+\cdots,
      x+\frac{15}{17}\,y-\frac{31}{17}]$,
    \item $R_{2\,S',y^2}=[\frac{1381}{189}\,x^4\,y^2+\cdots,
      y^2+\frac{106}{189}\,x -\frac{17}{63}\,y-\frac{134}{189}]$,
      $\tilde{R}_{2\,S',y^2}=[\frac{1381}{189}\,x^4\,y^2+\cdots,
      y^2+\frac{106}{189}\,x -\frac{17}{63}\,y-\frac{134}{189}]$,
    \item $R_{2\,S',x\,y}=[-4\,x^4\,y^2+\cdots,x\,y+x-y-1]$.
    \end{itemize}
    We
    initialize
    $R_{2\,S',x^2}=[-\frac{945}{34}\,x^4\,y^3-x^4\,y^3+\cdots,
    x^2+\frac{15}{17}\,x\,y-\frac{45}{34}\,x+\frac{15}{34}\,y-\frac{47}{17}]$
    using $\NormalFormHigherPart$ on $\tilde{R}_{2\,S',1}$ and $\tilde{R}_{2\,S',x}$
    and notice that the support of
    $C_{x^2}$ contains $x\,y$ which is not in $S'$.

    We then reduce it to
    $[-12\,x^3\,y^3+\cdots, x^2-1]$ by computing first the reduction
    of the right part by $x\,y+x-y-1$ and then by computing the
    reduction of the left parts without taking into account
    any monomials $\frac{\LCM(S')^2}{x\,y\,t}=\frac{x^3\,y^3}{t}$.

    As $\tilde{F}_{2\,S',x^2}=-x^4\,y^3+\cdots\neq F_{2\,S',x^2}$ and
    $\LM(\tilde{F}_{2\,S',x^2})=x^4\,y^3\prec\frac{\LCM(S')^2}{x^2}$ the relation
    succeeds!

    We update $R_{2\,S,x^2}=[-9\,x^2\,y^3+\cdots,x^2-1]$ and put
    it in $G$.

  \item We set $m=y^3$, $S'=\{1,y,x,y^2,y^3\}$,
    $R_{2\,S',B_1}=[x^3,0]$, $R_{2\,S',B_2}=[y^7,0]$,
    \begin{itemize}
    \item $R_{2\,S',1}=[6\,x^2\,y^6+\cdots,1]$, $\tilde{R}_{2\,S',1}=[6\,x^2\,y^6+\cdots,1]$,
    \item
      $R_{2\,S',y}=[-\frac{17}{2}\,x^2\,y^5-+\cdots,y-\frac{3}{2}]$,
      $\tilde{R}_{2\,S',y}=[-\frac{17}{2}\,x^2\,y^5-+\cdots,y-\frac{3}{2}]$,
    \item $R_{2\,S',x}=[\frac{189}{17}\,x\,y^6+\cdots,
      x+\frac{15}{17}\,y-\frac{31}{17}]$,
      $\tilde{R}_{2\,S',x}=[\frac{189}{17}\,x\,y^6+\cdots,
      x+\frac{15}{17}\,y-\frac{31}{17}]$,
    \item $R_{2\,S',y^2}=[\frac{1381}{189}\,x^2\,y^4+\cdots,
      y^2+\frac{106}{189}\,x -\frac{17}{63}\,y-\frac{134}{189}]$,
      $\tilde{R}_{2\,S',y^2}=[\frac{1381}{189}\,x^2\,y^4+\cdots,
      y^2+\frac{106}{189}\,x -\frac{17}{63}\,y-\frac{134}{189}]$,
    \item $R_{2\,S',x\,y}=[-9\,x\,y^6+\cdots,x\,y+x-y-1]$,
    \item $R_{2\,S',x^2}=[-9\,x^2\,y^5+\cdots,x^2-1]$.
    \end{itemize}
    We
    initialize
    $R_{2\,S',y^3}=[-\frac{4204}{1381}\,x\,y^5+\cdots,
    y^3+\frac{1354}{1381}\,y^2-\frac{607}{1381}\,x-\frac{190}{1381}\,y-\frac{770}{1381}]$
    using $\NormalFormHigherPart$ on $R_{2\,S',y}$ and $R_{2\,S',y^2}$.
  
    We then reduce it to
    $[-\frac{4204}{1381}\,x\,y^5-\frac{4620}{1381}\,y^6-\frac{25651}{1381}\,
    x^2\,y^3+\cdots,
    y^3+\frac{1354}{1381}\,y^2-\frac{607}{1381}\,x-\frac{190}{1381}\,y
    -\frac{770}{1381}]$ 
    by computing the
    reduction of the left parts without taking into account
    any monomials
    $\frac{\LCM(S')^2}{x\,y\,t}=\frac{x\,y^5}{t},
    \frac{\LCM(S')^2}{x^2\,t}=\frac{y^6}{t}$.

    As $\tilde{F}_{2\,S',y^3}=-\frac{25651}{1381}\,x^2\,y^3+\cdots\neq F_{2\,S',y^3}$ and
    $\LM(\tilde{F}_{2\,S',y^3})=x^2\,y^3=\frac{\LCM(S')^2}{y^3}$ the relation
    fails!

    $L$ is updated to $\{y^4\}$.
  \item We set $m=y^4$, $S'=\{1,y,x,y^2,y^3,y^4\}$,
    $R_{2\,S',B_1}=[x^3,0]$, $R_{2\,S',B_2}=[y^9,0]$,
    \begin{itemize}
    \item $R_{2\,S',1}=[6\,x^2\,y^8+\cdots,1]$, $\tilde{R}_{2\,S',1}=[6\,x^2\,y^8+\cdots,1]$,
    \item
      $R_{2\,S',y}=[-\frac{17}{2}\,x^2\,y^7-+\cdots,y-\frac{3}{2}]$,
      $\tilde{R}_{2\,S',y}=[-\frac{17}{2}\,x^2\,y^7-+\cdots,y-\frac{3}{2}]$,
    \item $R_{2\,S',x}=[\frac{189}{17}\,x\,y^8+\cdots,
      x+\frac{15}{17}\,y-\frac{31}{17}]$,
      $\tilde{R}_{2\,S',x}=[\frac{189}{17}\,x\,y^8+\cdots,
      x+\frac{15}{17}\,y-\frac{31}{17}]$,
    \item $R_{2\,S',y^2}=[\frac{1381}{189}\,x^2\,y^6+\cdots,
      y^2+\frac{106}{189}\,x -\frac{17}{63}\,y-\frac{134}{189}]$,
      $\tilde{R}_{2\,S',y^2}=[\frac{1381}{189}\,x^2\,y^6+\cdots,
      y^2+\frac{106}{189}\,x -\frac{17}{63}\,y-\frac{134}{189}]$,
    \item
      $R_{2\,S',y^3}=[\frac{5463}{1381}\,x\,y^7-\frac{4620}{1381}\,y^8
      -\frac{25651}{1381}\,x^2\,y^5+\cdots,
      y^3+\frac{1354}{1381}\,y^2-\frac{607}{1381}\,x-\frac{190}{1381}\,y
      -\frac{770}{1381}]$,
      $\tilde{R}_{2\,S',y^3}=[-\frac{25651}{1381}\,x^2\,y^5+\cdots,
      y^3+\frac{1354}{1381}\,y^2-\frac{607}{1381}\,x-\frac{190}{1381}\,y
      -\frac{770}{1381}]$,
    \item $R_{2\,S',x\,y}=[-9\,x\,y^8+\cdots,x\,y+x-y-1]$,
    \item $R_{2\,S',x^2}=[-9\,x^2\,y^7+\cdots,x^2-1]$.
    \end{itemize}
    We
    initialize
    $R_{2\,S',y^4}$ using $\NormalFormHigherPart$ on $R_{2\,S',y^2}$ and $R_{2\,S',y^3}$
    and then reduce it to
    $[-\frac{571416}{25651}\,x\,y^7+\cdots+\frac{504287}{25651}\,x^2\,y^4+\cdots,
    y^4-\frac{29900}{25651}\,y^3+\frac{7703}{25651}\,y^2+\frac{72041}{25651}\,x
    -\frac{35598}{25651}\,y-\frac{26811}{25651}]$
    by computing the
    reduction of the left parts without taking into account
    any monomials
    $\frac{\LCM(S')^2}{x\,y\,t}=\frac{x\,y^7}{t},
    \frac{\LCM(S')^2}{x^2\,t}=\frac{y^8}{t}$.

    As $\tilde{F}_{2\,S',y^4}=\frac{504287}{25651}\,x^2\,y^4+\cdots\neq F_{2\,S',y^3}$ and
    $\LM(\tilde{F}_{2\,S',y^4})=x^2\,y^4=\frac{\LCM(S')^2}{y^4}$ the relation
    fails!

    $L$ is updated to $\{y^5\}$.
  \item We set $m=y^5$, $S'=\{1,y,x,y^2,y^3,y^4,y^5\}$,
    $R_{2\,S',B_1}=[x^3,0]$, $R_{2\,S',B_2}=[y^{11},0]$,
    \begin{itemize}
    \item $R_{2\,S',1}=[6\,x^2\,y^{10}+\cdots,1]$,
      $\tilde{R}_{2\,S',1}=[6\,x^2\,y^{10}+\cdots,1]$,
    \item
      $R_{2\,S',y}=[-\frac{17}{2}\,x^2\,y^9-+\cdots,y-\frac{3}{2}]$,
      $\tilde{R}_{2\,S',y}=[-\frac{17}{2}\,x^2\,y^9-+\cdots,y-\frac{3}{2}]$,
    \item $R_{2\,S',x}=[\frac{189}{17}\,x\,y^{10}+\cdots,
      x+\frac{15}{17}\,y-\frac{31}{17}]$,
      $\tilde{R}_{2\,S',x}=[\frac{189}{17}\,x\,y^{10}+\cdots,
      x+\frac{15}{17}\,y-\frac{31}{17}]$,
    \item $R_{2\,S',y^2}=[\frac{1381}{189}\,x^2\,y^8+\cdots,
      y^2+\frac{106}{189}\,x -\frac{17}{63}\,y-\frac{134}{189}]$,
      $\tilde{R}_{2\,S',y^2}=[\frac{1381}{189}\,x^2\,y^8+\cdots,
      y^2+\frac{106}{189}\,x -\frac{17}{63}\,y-\frac{134}{189}]$,
    \item
      $R_{2\,S',y^3}=[\frac{5463}{1381}\,x\,y^9-\frac{4620}{1381}\,y^{10}
      -\frac{25651}{1381}\,x^2\,y^7+\cdots,
      y^3+\frac{1354}{1381}\,y^2-\frac{607}{1381}\,x-\frac{190}{1381}\,y
      -\frac{770}{1381}]$,
      $\tilde{R}_{2\,S',y^3}=[-\frac{25651}{1381}\,x^2\,y^7+\cdots,
      y^3+\frac{1354}{1381}\,y^2-\frac{607}{1381}\,x-\frac{190}{1381}\,y
      -\frac{770}{1381}]$,
    \item
      $R_{2\,S',y^4}=[-\frac{648369}{25651}\,x\,y^9+\cdots+\frac{504287}{25651}\,x^2\,y^6
      +\cdots,y^4-\frac{29900}{25651}\,y^3+\frac{7703}{25651}\,y^2+\frac{72041}{25651}\,x
      -\frac{35598}{25651}\,y-\frac{26811}{25651}]$,
      $\tilde{R}_{2\,S',y^4}=[\frac{504287}{25651}\,x^2\,y^6
      +\cdots,y^4-\frac{29900}{25651}\,y^3+\frac{7703}{25651}\,y^2+\frac{72041}{25651}\,x
      -\frac{35598}{25651}\,y-\frac{26811}{25651}]$,
    \item $R_{2\,S',x\,y}=[-9\,x\,y^{10}+\cdots,x\,y+x-y-1]$,
    \item $R_{2\,S',x^2}=[-9\,x^2\,y^9+\cdots,x^2-1]$.
    \end{itemize}
    We
    initialize
    $R_{2\,S',y^5}$ using $\NormalFormHigherPart$ on $R_{2\,S',y^3}$ and $R_{2\,S',y^4}$
    and then reduce it to
    $[12\,x\,y^9+6\,y^{10}+2\,x\,y^8+4\,x\,y^7+7\,x\,y^6-9\,x^2\,y^4-3\,x\,y^5
    -5\,x^2\,y^3-x^2\,y^2-10\,x^2\,y+6\,x^2, y^5+1]$
    by computing the
    reduction of the left parts without taking into account
    any monomials
    $\frac{\LCM(S')^2}{x\,y\,t}=\frac{x\,y^9}{t},
    \frac{\LCM(S')^2}{x^2\,t}=\frac{y^{10}}{t}$.

    As $\tilde{F}_{2\,S',y^5}=-9\,x^2\,y^4+\cdots\neq F_{2\,S',y^3}$ and
    $\LM(\tilde{F}_{2\,S',y^5})=x^2\,y^4\prec\frac{\LCM(S')^2}{y^5}$ the relation
    succeeds!

    We update $R_{2\,S,y^5}=[3\,x\,y^8+\cdots,y^5+1]$ and put
    it in $G$.
  \end{enumerate}
  The algorithms returns $G$, in particular the second part of each
  pair: $[x\,y+x-y-1,x^2-1,y^5+1]$.
\end{example}

\begin{remark}
  At step $y^4$, updating $R_{2\,S,y^3}$ into $R_{2\,S',y^3}$ makes
  the leading monomial of $F_{2\,S',y^3}$,
  $\frac{5463}{1381}\,\frac{\LCM(S')^2}{x\,y}$, totally different from
  this of $F_{2\,S,y^3}$,
  $-\frac{4204}{1381}\,\frac{\LCM(S)^2}{x\,y}$. 
  However, this is not the case when considering the leading
  monomials of $\tilde{F}_{2\,S',y^3}$ and $\tilde{F}_{2\,S,y^3}$, as
  $\LM(\tilde{F}_{2\,S',y^3})=-\frac{25651}{1381}\,\frac{\LCM(S')^2}{y^3}
  =\LM(\tilde{F}_{2\,S,y^3})\,\frac{\LCM(S')^2}{\LCM(S)^2}$.
\end{remark}


%% file: 6-bench.tex
In this section, we report on the number of counted arithmetic
operations performed by
the different algorithms for computing the \gb of the ideal of
relations of some table families. They are counted using naive
multiplications. To do so, we have a counter that is
incremented by $1$ each time a product in the base field is
performed. For advanced functions that were not recoded, we use a
formula. For instance, after a call to
$\NormalForm(F_m,[F_1,\ldots,F_r])$, assuming the
associated quotients are $Q_1,\ldots,Q_r$, then we use the following
formula to increment the counter
\[
  \Card{\supp Q_1}\Card{\supp F_1}+\cdots+\Card{\supp Q_r}\Card{\supp F_r}.
\]

Our implementation is available at
\url{https://www-polsys.lip6.fr/~berthomieu/Guessing.mpl}.

For the \sFGLM algorithm, this corresponds to the
number of multiplications in the Gaussian elimination step to determine
the staircase and the number of multiplications to solve the
triangular systems, for each polynomial in the output \gb, to compute
its coefficients.

For the \BMS algorithm, this corresponds to the
number of multiplications for testing each relation at a monomial and
then updating the set of relations by correcting some of them or
building new ones.

For the \AGbb algorithm, this corresponds to the number of
multiplications for computing the modified Gram-Schmidt
orthogonalization process and the evaluation of polynomials on the
table.

For the \divalgo algorithm, this corresponds to the number of
multiplications to obtain the partial quotients in the normal forms
computations and the polynomials $R_m$'s.

Three families in dimension $2$
(Figure~\ref{fig:basicop2D})
and dimension $3$ (Figure~\ref{fig:basicop3D}) are tested. For each of
them we use the
$\DRL(z\prec y\prec x)$ ordering and denote by $S$ the staircase and
$\LM(\cG)$ the set of the leading monomials of the \gb of relations.
\begin{description}
\item[Rectangle tables:]
  $\LM(\cG)=\big\{y^{\floor{\sfrac{d}{2}}},x^d\big\}$ in dimension $2$ and
  $\LM(\cG)=\big\{z^{\ceil{\sfrac{d}{3}}},y^{\floor{\sfrac{d}{2}}},x^d\big\}$
  dimension $3$. This case is the best for the size of
  the \gb compared to the size of the staircase.
\item[\textsc{L}-shape tables:] $\LM(\cG)=\big\{x\,y,y^d,x^d\big\}$ in dimension $2$ and
  $\LM(\cG)=\big\{y\,z,x\,z,x\,y,z^d,y^d,x^d\big\}$ in dimension
  $3$. This case is the worst for the number of table
  queries compared to the sizes of the staircase
  and the \gb.
\item[Simplex tables:] 
  $\LM(\cG)=\big\{y^d,x\,y^{d-1},\ldots,x^d\big\}$ in dimension
  $2$ and
  $\LM(\cG)=\big\{z^d,y\,z^{d-1},x\,z^{d-1},\ldots,y^d,\allowbreak
  x\,y^{d-1},\ldots,x^d\big\}$
  in dimension $3$, \ie all the monomials of degree $d$. This
  case is the best for the number of table queries and the worst for the size of
  the \gb, both compared to the size of the staircase.
\end{description}

Let $a=\max(S\cup\LM(\cG))$. Generically, a relation $C_m$
fails when shifted by $m$. From~\cite[Prop.~10]{Bras-Amoros2006}, we
know that the \BMS algorithm recover all the relations when called
up to monomial 
$\max(S)\,\max(S\cup\LM(\cG))$. Yet, if $\max(\LM(\cG))\succ\max(S)$,
then for $g\in\LM(\cG)$, the relation $C_g$ is not necessarily shifted
by $g$, so we called it with $a^2$. So was the \AGbb algorithm.
The \sFGLM algorithm
was called on $U=T=\cT[a]$. The \divalgo algorithm was called on
$U=\{1\},T=\cT[a^2]$ and $U=T=\cT[a]$ and we report the lower number
of operations.

\begin{figure*}[htbp!]
  \centering
  \pgfplotsset{
    small,
    width=12cm,
    height=8cm,
    legend cell align=left,
    legend columns=4,
    legend style={at={(-0.05,0.94)},anchor=south
      west,font=\scriptsize,
    }
  }
  \begin{tikzpicture}[baseline]
    \begin{axis}[
      ymode=log,
      xlabel={$d$},
      xlabel style={at={(0.95,0.1)}},
      xmin=4.8,xmax=30.2,
      ymin=15,ymax=2200,
      xtick={5,10,15,20,25,30},
      ytick={0.1,0.2,0.3,0.4,0.5,0.6,0.7,0.8,0.9,1,2,3,4,5,6,7,8,9,
        10,20,30,40,50,60,70,80,90,100,200,300,400,500,600,700,800,900,
        1000,2000,3000,4000,5000,6000,7000,8000,9000,10000,20000},
      yticklabels={},
      extra y ticks={0.1,0.5,1,5,10,50,100,500,1000,5000,10000},
      extra y tick labels={0.1,0.5,1,5,10,50,100,500,1000,5000,10000},
      ylabel={\#\,Arithm. Op/(\#\,S)$^2$},
      ylabel style={at={(0.08,0.75)}},
      ]
      \addlegendimage{empty legend}
      \addlegendentry{}
      \addlegendimage{legend image with text=Rectangle}
      \addlegendentry{}
      \addlegendimage{legend image with text=\textsc{L} shape}
      \addlegendentry{}
      \addlegendimage{legend image with text=Simplex}
      \addlegendentry{}

      \addlegendimage{empty legend}
      \addlegendentry{\sFGLM}
      \addplot[thick,every mark/.append style={solid},
      mark=triangle*,mark phase=5,mark repeat=6,dashed,red] plot coordinates {
        (5,2850/10^2) (6,13789/18^2) (7,27307/21^2) (8,85407/32^2) 
        (9,142937/36^2) (10,348250/50^2) (11,525580/55^2) (12,1090262/72^2) 
        (13,1535195/78^2) (14,2847334/98^2) (15,3815490/105^2) (16,6517514/128^2) 
        (17,8415422/136^2) (18,13497087/162^2) (19,16932567/171^2) (20,25838525/200^2) 
        (21,31678360/210^2) (22,46430307/242^2) (23,55865205/253^2) (24,79198609/288^2) 
        (25,93815455/300^2) (26,129330864/338^2) (27,151192262/351^2)
        (28,203521192/392^2) 
        (29,235252297/406^2) (30,310237700/450^2) 
      };
      \addlegendentry{}
      \addplot[thick,every mark/.append style={solid},
      mark=triangle*,mark phase=3,mark repeat=6,dashed,blue] plot coordinates {
        (5,3329/9^2) (6,7623/11^2) (7,15885/13^2) (8,30675/15^2) 
        (9,55638/17^2) (10,95774/19^2) (11,157738/21^2) (12,250170/23^2) 
        (13,384055/25^2) (14,573113/27^2) (15,834219/29^2) (16,1187853/31^2) 
        (17,1658580/33^2) (18,2275560/35^2) (19,3073088/37^2) (20,4091164/39^2) 
        (21,5376093/41^2) (22,6981115/43^2) (23,8967065/45^2) (24,11403063/47^2) 
        (25,14367234/49^2) (26,17947458/51^2) (27,22242150/53^2) (28,27361070/55^2) 
        (29,33426163/57^2) (30,40572429/59^2) 
      };
      \addlegendentry{}
      \addplot[thick,every mark/.append style={solid},
      mark=triangle*,mark phase=1,mark repeat=6,dashed,green] plot coordinates {
        (5,5480/15^2) (6,12957/21^2) (7,27230/28^2) (8,52374/36^2) 
        (9,94005/45^2) (10,159610/55^2) (11,258907/66^2) (12,404235/78^2) 
        (13,610974/91^2) (14,897995/105^2) (15,1288140/120^2) (16,1808732/136^2) 
        (17,2492115/153^2) (18,3376224/171^2) (19,4505185/190^2) (20,5929945/210^2) 
        (21,7708932/231^2) (22,9908745/253^2) (23,12604874/276^2) (24,15882450/300^2) 
        (25,19837025/325^2) (26,24575382/351^2) (27,30216375/378^2) (28,36891799/406^2) 
        (29,44747290/435^2) (30,53943255/465^2)      
      };
      \addlegendentry{}

      \addlegendimage{empty legend}
      \addlegendentry{\BMS}
      \addplot[thick,every mark/.append style={solid,rotate=180},
      mark=triangle*,mark phase=5,mark repeat=6,dotted,red] plot coordinates {
        (5,5148/10^2) (6,18677/18^2) (7,27720/21^2) (8,69808/32^2) 
        (9,94593/36^2) (10,198745/50^2) (11,254459/55^2) (12,473240/72^2) 
        (13,582804/78^2) (14,992697/98^2) (15,1188426/105^2) (16,1894408/128^2) 
        (17,2219615/136^2) (18,3359849/162^2) (19,3870240/171^2) (20,5620354/200^2) 
        (21,6386462/210^2) (22,8964369/242^2) (23,10071840/253^2) (24,13742132/288^2) 
        (25,15294831/300^2) (26,20373282/338^2) (27,22493216/351^2) (28,29352546/392^2) 
        (29,32184023/406^2) (30,41255234/450^2) 
      };
      \addlegendentry{}
      \addplot[thick,every mark/.append style={solid,rotate=180},
      mark=triangle*,mark phase=3,mark repeat=6,dotted,blue] plot coordinates {
        (5,6630/9^2) (6,10844/11^2) (7,16540/13^2) (8,23942/15^2) 
        (9,33274/17^2) (10,44760/19^2) (11,58624/21^2) (12,75090/23^2) 
        (13,94382/25^2) (14,116724/27^2) (15,142340/29^2) (16,171454/31^2) 
        (17,204290/33^2) (18,241072/35^2) (19,282024/37^2) (20,327370/39^2) 
        (21,377334/41^2) (22,432140/43^2) (23,492012/45^2) (24,557174/47^2) 
        (25,627850/49^2) (26,704264/51^2) (27,786640/53^2) (28,875202/55^2) 
        (29,970174/57^2) (30,1071780/59^2) 
      };
      \addlegendentry{}
      \addplot[thick,every mark/.append style={solid,rotate=180},
      mark=triangle*,mark phase=1,mark repeat=6,dotted,green] plot coordinates {
        (5,15323/15^2) (6,31823/21^2) (7,60325/28^2) (8,106402/36^2) 
        (9,177215/45^2) (10,281695/55^2) (11,430671/66^2) (12,637120/78^2) 
        (13,916343/91^2) (14,1286155/105^2) (15,1767041/120^2) (16,2382458/136^2) 
        (17,3158859/153^2) (18,4126035/171^2) (19,5317216/190^2) (20,6769342/210^2) 
        (21,8523177/231^2) (22,10623524/253^2) (23,13119506/276^2) (24,16064408/300^2) 
        (25,19516475/325^2) (26,23538737/351^2) (27,28198971/378^2) (28,33570439/406^2) 
        (29,39731608/435^2) (30,46766425/465^2) 
      };
      \addlegendentry{}

      \addlegendimage{empty legend}
      \addlegendentry{\AGbb}
      \addplot[thick,every mark/.append style={solid},
      mark=*, mark phase=5,mark repeat=6,densely dotted,red] plot coordinates {
        (5,1126/10^2) (6,3938/18^2) (7,6452/21^2) (8,17536/32^2) 
        (9,26040/36^2) (10,60216/50^2) (11,83645/55^2) (12,171103/72^2) 
        (13,226350/78^2) (14,421075/98^2) (15,536819/105^2) (16,926930/128^2) 
        (17,1148231/136^2) (18,1868124/162^2) (19,2261778/171^2) (20,3507384/200^2) 
        (21,4168034/210^2) (22,6210112/242^2) (23,7266974/253^2) (24,10477022/288^2) 
        (25,12101830/300^2) (26,16956515/338^2) (27,19371978/351^2) (28,26481149/392^2) 
        (29,29968940/406^2) (30,40105002/450^2) 
      };
      \addlegendentry{}
      \addplot[thick,every mark/.append style={solid},
      mark=*, mark phase=3,mark repeat=6,densely dotted,blue] plot coordinates {
        (5,1499/9^2) (6,2696/11^2) (7,4446/13^2) (8,6882/15^2) 
        (9,10149/17^2) (10,14404/19^2) (11,19816/21^2) (12,26566/23^2) 
        (13,34847/25^2) (14,44864/27^2) (15,56834/29^2) (16,70986/31^2) 
        (17,87561/33^2) (18,106812/35^2) (19,129004/37^2) (20,154414/39^2) 
        (21,183331/41^2) (22,216056/43^2) (23,252902/45^2) (24,294193/47^2) 
        (25,340269/49^2) (26,391475/51^2) (27,448176/53^2) (28,510741/55^2) 
        (29,579559/57^2) (30,655024/59^2) 
      };
      \addlegendentry{}
      \addplot[thick,every mark/.append style={solid},
      mark=*, mark phase=1,mark repeat=6,densely dotted,green] plot coordinates {
        (5,2273/15^2) (6,5454/21^2) (7,12012/28^2) (8,23534/36^2) 
        (9,43960/45^2) (10,77477/55^2) (11,130088/66^2) (12,209639/78^2) 
        (13,325838/91^2) (14,491146/105^2) (15,720420/120^2) (16,1036051/136^2) 
        (17,1461141/153^2) (18,2013399/171^2) (19,2741189/190^2) (20,3663086/210^2) 
        (21,4846435/231^2) (22,6313124/253^2) (23,8155865/276^2) (24,10398550/300^2) 
        (25,13166427/325^2) (26,16483793/351^2) (27,20514716/378^2) (28,25281752/406^2) 
        (29,30999775/435^2) (30,37684447/465^2) 
      };
      \addlegendentry{}

      \addlegendimage{empty legend}
      \addlegendentry{\divalgo}
      \addplot[thick,every mark/.append style={solid},
      mark=square*,mark phase=5,mark repeat=6,red] plot coordinates {
        (5,1942/10^2) (6,5869/18^2) (7,8251/21^2) (8,19181/32^2) 
        (9,25071/36^2) (10,49382/50^2) (11,61541/55^2) (12,108492/72^2) 
        (13,130770/78^2) (14,212964/98^2) (15,250465/105^2) (16,384335/128^2) 
        (17,443581/136^2) (18,649879/162^2) (19,738979/171^2) (20,1043251/200^2) 
        (21,1172057/210^2) (22,1605130/242^2) (23,1785421/253^2) (24,2383893/288^2) 
        (25,2629514/300^2) (26,3436233/338^2) (27,3763273/351^2) (28,4827838/392^2) 
        (29,5254802/406^2) (30,6633996/450^2) 
      };
      \addlegendentry{}
      \addplot[thick,every mark/.append style={solid},
      mark=square*,mark phase=3,mark repeat=6,blue] plot coordinates {
        (5,2624/9^2) (6,3913/11^2) (7,5550/13^2) (8,7575/15^2) 
        (9,10028/17^2) (10,12949/19^2) (11,16378/21^2) (12,20355/23^2) 
        (13,24920/25^2) (14,30113/27^2) (15,35974/29^2) (16,42543/31^2) 
        (17,49860/33^2) (18,57963/35^2) (19,66898/37^2) (20,76699/39^2) 
        (21,87408/41^2) (22,99065/43^2) (23,111710/45^2) (24,125383/47^2) 
        (25,140124/49^2) (26,155973/51^2) (27,172970/53^2) (28,191155/55^2) 
        (29,210568/57^2) (30,231249/59^2) 
      };
      \addlegendentry{}
      \addplot[thick,every mark/.append style={solid},
      mark=square*,mark phase=1,mark repeat=6,green] plot coordinates {
        (5,9370/15^2) (6,18657/21^2) (7,34153/28^2) (8,58573/36^2) 
        (9,95317/45^2) (10,148561/55^2) (11,223344/66^2) (12,325628/78^2) 
        (13,462416/91^2) (14,641770/105^2) (15,872941/120^2) (16,1166516/136^2) 
        (17,1534269/153^2) (18,1989491/171^2) (19,2546948/190^2) (20,3222970/210^2) 
        (21,4035554/231^2) (22,5004434/253^2) (23,6151155/276^2) (24,7499169/300^2) 
        (25,9073894/325^2) (26,10902801/351^2) (27,13015490/378^2) (28,15443771/406^2) 
        (29,18221808/435^2) (30,21386093/465^2)
      };
      \addlegendentry{}

    \end{axis}
  \end{tikzpicture}
  \caption{Number of arithmetic operations (\textsc{2D})}
  \label{fig:basicop2D}
\end{figure*}
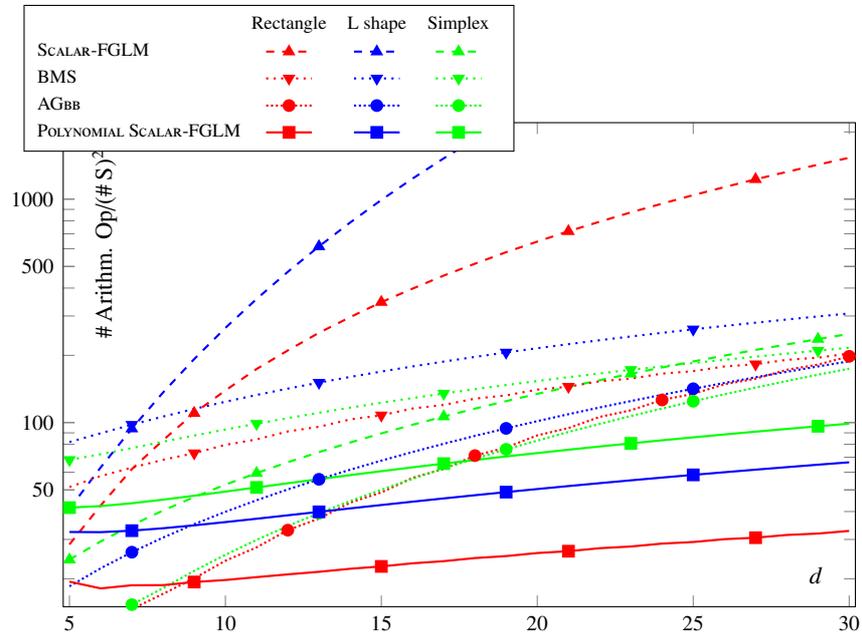
\begin{figure*}[htbp!]
  \centering
  \pgfplotsset{
    small,
    width=12cm,
    height=8cm,
    legend cell align=left,
    legend columns=4,
    legend style={at={(-0.05,0.94)},anchor=south
      west,font=\scriptsize,
    }
  }
  \begin{tikzpicture}[baseline]
    \begin{axis}[
      ymode=log,
      xlabel={$d$},
      xlabel style={at={(0.95,0.1)}},
      xmin=4.8,xmax=15.2,
      ymin=35,ymax=8800,
      xtick={1,...,25},
      ytick={0.1,0.2,0.3,0.4,0.5,0.6,0.7,0.8,0.9,1,2,3,4,5,6,7,8,9,
        10,20,30,40,50,60,70,80,90,100,200,300,400,500,600,700,800,900,
        1000,2000,3000,4000,5000,6000,7000,8000,9000,10000,20000},
      yticklabels={},
      extra y ticks={0.1,0.5,1,5,10,50,100,500,1000,5000,10000},
      extra y tick labels={0.1,0.5,1,5,10,50,100,500,1000,5000,10000},
      ylabel={\#\,Arithm. Op/(\#\,S)$^2$},
      ylabel style={at={(0.08,0.75)}},
      ]
      \addlegendimage{empty legend}
      \addlegendentry{}
      \addlegendimage{legend image with text=Rectangle}
      \addlegendentry{}
      \addlegendimage{legend image with text=\textsc{L} shape}
      \addlegendentry{}
      \addlegendimage{legend image with text=Simplex}
      \addlegendentry{}

      \addlegendimage{empty legend}
      \addlegendentry{\sFGLM}
      \addplot[thick,every mark/.append style={solid},
      mark=triangle*,mark phase=3,mark repeat=3,dashed,red] plot coordinates {
        (5,145790/20^2) (6,1213257/36^2) (7,6052900/63^2) (8,25915409/96^2) 
        (9,49554965/108^2) (10,256850846/200^2) (11,425464350/220^2)
        (12,1075627255/288^2) 
      };
      \addlegendentry{}
      \addplot[thick,every mark/.append style={solid},
      mark=triangle*,mark phase=2,mark repeat=3,dashed,blue] plot coordinates {
        (5,58930/13^2) (6,197030/16^2) (7,573038/19^2) (8,1489466/22^2) 
        (9,3532520/25^2) (10,7766275/28^2) (11,16021340/31^2) (12,31309023/34^2)
        (13,58398166/37^2) (14,104598540/40^2) (15,180801970/43^2) 
      };
      \addlegendentry{}
      \addplot[thick,every mark/.append style={solid},
      mark=triangle*,mark phase=1,mark repeat=3,dashed,green] plot coordinates {
        (5,107695/35^2) (6,368102/56^2) (7,1073828/84^2) (8,2774390/120^2) 
        (9,6510845/165^2) (10,14131315/220^2) (11,28747862/286^2)
        (12,55376503/364^2)
      };
      \addlegendentry{}

      \addlegendimage{empty legend}
      \addlegendentry{\BMS}
      \addplot[thick,every mark/.append style={solid,rotate=180},
      mark=triangle*,mark phase=3,mark repeat=3,dotted,red] plot coordinates {
        (5,113816/20^2) (6,477598/36^2) (7,1820818/63^2) (8,5275639/96^2) 
        (9,7677600/108^2) (10,32891994/200^2) (11,44478509/220^2) (12,90407631/288^2) 
      };
      \addlegendentry{}
      \addplot[thick,every mark/.append style={solid,rotate=180},
      mark=triangle*,mark phase=2,mark repeat=3,dotted,blue] plot coordinates {
        (5,97491/13^2) (6,198434/16^2) (7,364700/19^2) (8,621569/22^2) 
        (9,999469/25^2) (10,1534360/28^2) (11,2268118/31^2) (12,3248919/34^2) 
        (13,4531623/37^2) (14,6178158/40^2) (15,8257904/43^2)      
      };
      \addlegendentry{}
      \addplot[thick,every mark/.append style={solid,rotate=180},
      mark=triangle*,mark phase=1,mark repeat=3,dotted,green] plot coordinates {
        (5,380110/35^2) (6,1169668/56^2) (7,3147889/84^2) (8,7615977/120^2) 
        (9,16908231/165^2) (10,34966045/220^2) (11,68139075/286^2) (12,126240493/364^2)
      };
      \addlegendentry{}

      \addlegendimage{empty legend}
      \addlegendentry{\AGbb}
      \addplot[thick,every mark/.append style={solid},
      mark=*,mark phase=3,mark repeat=3,densely dotted,red] plot coordinates {
        (5,26614/20^2) (6,111057/36^2) (7,381888/63^2) (8,1122216/96^2) 
        (9,1696681/108^2) (10,7525329/200^2) (11,10637723/220^2) (12,23062952/288^2)
      };
      \addlegendentry{}
      \addplot[thick,every mark/.append style={solid},
      mark=*, mark phase=3,mark repeat=6,densely dotted,blue] plot coordinates {
        (5,15038/13^2) (6,31001/16^2) (7,57246/19^2) (8,97596/22^2) 
        (9,156572/25^2) (10,239429/28^2) (11,352192/31^2) (12,501691/34^2) 
        (13,695602/37^2) (14,942473/40^2) (15,1251769/43^2)
      };
      \addlegendentry{}
      \addplot[thick,every mark/.append style={solid},
      mark=*, mark phase=1,mark repeat=6,densely dotted,green] plot coordinates {
        (5,44833/35^2) (6,157733/56^2) (7,477786/84^2) (8,1278889/120^2) 
        (9,3098535/165^2) (10,6927441/220^2) (11,14415281/286^2) (12,28454798/364^2) 
      };
      \addlegendentry{}

      \addlegendimage{empty legend}
      \addlegendentry{\divalgo}
      \addplot[thick,every mark/.append style={solid},
      mark=square*,mark phase=3,mark repeat=3,red] plot coordinates {
        (5,17547/20^2) (6,65196/36^2) (7,218286/63^2) (8,580542/96^2) 
        (9,799578/108^2) (10,3186694/200^2)
        (11,4146968/220^2) (12,7977732/288^2) 
        (13,16120801/390^2) (14,28306397/490^2) (15,34160003/525^2)
      };
      \addlegendentry{}
      \addplot[thick,every mark/.append style={solid},
      mark=square*,mark phase=2,mark repeat=3,blue] plot coordinates {
        (5,12131/13^2) (6,19919/16^2) (7,31000/19^2) (8,46210/22^2) 
        (9,66481/25^2) (10,92841/28^2) (11,126414/31^2) (12,168420/34^2) 
        (13,220175/37^2) (14,283091/40^2) (15,358676/43^2)
      };
      \addlegendentry{}
      \addplot[thick,every mark/.append style={solid},
      mark=square*,mark phase=1,mark repeat=3,green] plot coordinates {
        (5,201521/35^2) (6,616688/56^2) (7,1657736/84^2) (8,4019618/120^2) 
        (9,8950600/165^2) (10,18568195/220^2) (11,36288873/286^2)
        (12,67403361/364^2)
      };
      \addlegendentry{}

    \end{axis}
  \end{tikzpicture}
  \caption{Number of arithmetic operations (\textsc{3D})}
  \label{fig:basicop3D}
\end{figure*}
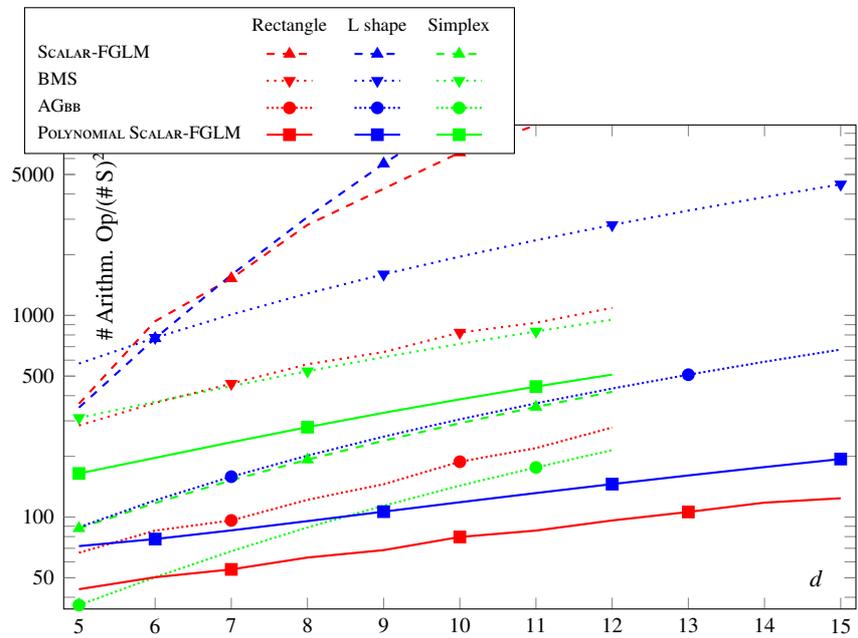

The \divalgo algorithm performs fewer arithmetic operations
than the others, for large $d$. More precisely, its
number of operations 
appears to be linear in
$(\Card{S})^2=O(\Card{S}\,(\Card{S}+\Card{\cG}))$ in fixed dimension.
\begin{description}
\item[Simplex tables:] While it seems the \AGbb and \sFGLM algorithms are the fastest
  in Figure~\ref{fig:basicop3D}, we can expect that it will not be
  the case in higher degrees where the \divalgo will be the
  fastest. This phenomenon is already observed in
  Figure~\ref{fig:basicop2D} in low degree. This
  would confirm the observed speedup in dimension $2$ to also dimension $3$.
\item[\textsc{L}-shape tables:] Although the obtained speedups
  are not negligible, the adaptive variant
  should allow us to perform even fewer operations. See Section~\ref{s:adaptive}.
\item[Rectangle tables:] While this family has the best behavior for
  the \BMS algorithm, the \divalgo algorithm has an even greater
  speedup than in the Simplex case.
\end{description}

In the following Figures~\ref{fig:basicopAdapt2D}
and~\ref{fig:basicopAdapt3D}, we compare the number of operations
performed by the \divalgo, the \asFGLM and the \Adivalgo algorithms on
these three families. First, as we can expect, the \Adivalgo
algorithms always performs fewer operations than the \divalgo
algorithm. This happens even when giving the sharper sets of monomials $\cT[a]$ and
$\cT[b]$ for the \divalgo algorithm. Furthermore, like in the
non-adaptive case, the polynomial
arithmetic seems to be faster than the linear algebra ones. Though, it
would be interesting to compare with a less naive implementation of the
\asFGLM algorithm.

That being said, in three variables, the \asFGLM algorithm is the
better one for the \textsc{L}-shape family in small size, until the
\Adivalgo becomes the better one for bigger sizes. As the crossing
happens later in \textsc{3D} than in \textsc{2D}, it is possible that
the more variables, the bigger the degrees need to be for the
\Adivalgo algorithm to be better than the \asFGLM algorithm, for this
family. This makes the \asFGLM algorithm still a suited algorithm in
this context.

\begin{figure*}[htbp!]
  \centering
  \pgfplotsset{
    small,
    width=12cm,
    height=8cm,
    legend cell align=left,
    legend columns=4,
    legend style={at={(-0.05,0.94)},anchor=south
      west,font=\scriptsize,
    }
  }
  \begin{tikzpicture}[baseline]
    \begin{axis}[
      ymode=log,
      xlabel={$d$},
      xlabel style={at={(0.95,0.1)}},
      xmin=8.8,xmax=50.2,
      ymin=1.5,ymax=220,
      xtick={5,10,15,20,25,30,35,40,45,50},
      ytick={0.1,0.2,0.3,0.4,0.5,0.6,0.7,0.8,0.9,1,2,3,4,5,6,7,8,9,
        10,20,30,40,50,60,70,80,90,100,200,300,400,500,600,700,800,900,
        1000,2000,3000,4000,5000,6000,7000,8000,9000,10000,20000},
      yticklabels={},
      extra y ticks={0.1,0.5,1,5,10,50,100,500,1000,5000,10000},
      extra y tick labels={0.1,0.5,1,5,10,50,100,500,1000,5000,10000},
      ylabel={\#\,Arithm. Op/(\#\,S)$^2$},
      ylabel style={at={(0.08,0.75)}},
      ]
      \addlegendimage{empty legend}
      \addlegendentry{}
      \addlegendimage{legend image with text=Rectangle}
      \addlegendentry{}
      \addlegendimage{legend image with text=\textsc{L} shape}
      \addlegendentry{}
      \addlegendimage{legend image with text=Simplex}
      \addlegendentry{}

      \addlegendimage{empty legend}
      \addlegendentry{\divalgo}
      \addplot[thick,every mark/.append style={solid},
      mark=square*,mark phase=9,mark repeat=12,red] plot coordinates {
        (5,1942/10^2) (6,5869/18^2) (7,8251/21^2) (8,19181/32^2) 
        (9,25071/36^2) (10,49382/50^2) (11,61541/55^2) (12,108492/72^2) 
        (13,130770/78^2) (14,212964/98^2) (15,250465/105^2) (16,384335/128^2) 
        (17,443581/136^2) (18,649879/162^2) (19,738979/171^2) (20,1043251/200^2) 
        (21,1172057/210^2) (22,1605130/242^2) (23,1785421/253^2) (24,2383893/288^2) 
        (25,2629514/300^2) (26,3436233/338^2) (27,3763273/351^2) (28,4827838/392^2) 
        (29,5254802/406^2) (30,6633996/450^2) 
      };
      \addlegendentry{}
      \addplot[thick,every mark/.append style={solid},
      mark=square*,mark phase=5,mark repeat=12,blue] plot coordinates {
        (5,2624/9^2) (6,3913/11^2) (7,5550/13^2) (8,7575/15^2) 
        (9,10028/17^2) (10,12949/19^2) (11,16378/21^2) (12,20355/23^2) 
        (13,24920/25^2) (14,30113/27^2) (15,35974/29^2) (16,42543/31^2) 
        (17,49860/33^2) (18,57963/35^2) (19,66898/37^2) (20,76699/39^2) 
        (21,87408/41^2) (22,99065/43^2) (23,111710/45^2) (24,125383/47^2) 
        (25,140124/49^2) (26,155973/51^2) (27,172970/53^2) (28,191155/55^2) 
        (29,210568/57^2) (30,231249/59^2) 
      };
      \addlegendentry{}
      \addplot[thick,every mark/.append style={solid},
      mark=square*,mark phase=1,mark repeat=12,green] plot coordinates {
        (5,9370/15^2) (6,18657/21^2) (7,34153/28^2) (8,58573/36^2) 
        (9,95317/45^2) (10,148561/55^2) (11,223344/66^2) (12,325628/78^2) 
        (13,462416/91^2) (14,641770/105^2) (15,872941/120^2) (16,1166516/136^2) 
        (17,1534269/153^2) (18,1989491/171^2) (19,2546948/190^2) (20,3222970/210^2) 
        (21,4035554/231^2) (22,5004434/253^2) (23,6151155/276^2) (24,7499169/300^2) 
        (25,9073894/325^2) (26,10902801/351^2) (27,13015490/378^2) (28,15443771/406^2) 
        (29,18221808/435^2) (30,21386093/465^2)
      };
      \addlegendentry{}

      \addlegendimage{empty legend}
      \addlegendentry{\asFGLM}
      \addplot[thick,every mark/.append style={solid},
      mark=triangle*,mark phase=9,mark repeat=12,loosely dashed,red] plot coordinates {
        (5,337/10^2) (6,1595/18^2) (7,2355/21^2) (8,7147/32^2) 
        (9,9871/36^2) (10,24599/50^2) (11,32254/55^2) (12,69587/72^2) 
        (13,87739/78^2) (14,170015/98^2) (15,208053/105^2) (16,371335/128^2) 
        (17,443935/136^2) (18,742787/162^2) (19,871616/171^2) (20,1384599/200^2) 
        (21,1600259/210^2) (22,2436147/242^2) (23,2780359/253^2) (24,4085075/288^2) 
        (25,4613103/300^2) (26,6577375/338^2) (27,7360690/351^2) (28,10228427/392^2) 
        (29,11357611/406^2) (30,15434999/450^2) (31,17022889/465^2) (32,22688207/512^2) 
        (33,24873279/528^2) (34,32587435/578^2) (35,35537428/595^2) (36,45855215/648^2) 
        (37,49770995/666^2) (38,63353067/722^2) (39,68472731/741^2) (40,86098299/800^2) 
        (41,92701519/820^2) (42,115281767/882^2) (43,123694374/903^2)
        (44,152286595/968^2) 
        (45,162885403/990^2) (46,198707855/1058^2)
        (47,211925725/1081^2) (48,256373207/1152^2) 
        (49,272704351/1176^2) (50,327364499/1250^2) 
      };
      \addlegendentry{}
      \addplot[thick,every mark/.append style={solid},
      mark=triangle*,mark phase=5,mark repeat=12,loosely dashed,blue] plot coordinates {
        (5,342/9^2) (6,517/11^2) (7,748/13^2) (8,1043/15^2) 
        (9,1410/17^2) (10,1857/19^2) (11,2392/21^2) (12,3023/23^2) 
        (13,3758/25^2) (14,4605/27^2) (15,5572/29^2) (16,6667/31^2) 
        (17,7898/33^2) (18,9273/35^2) (19,10800/37^2) (20,12487/39^2) 
        (21,14342/41^2) (22,16373/43^2) (23,18588/45^2) (24,20995/47^2) 
        (25,23602/49^2) (26,26417/51^2) (27,29448/53^2) (28,32703/55^2) 
        (29,36190/57^2) (30,39917/59^2) (31,43892/61^2) (32,48123/63^2) 
        (33,52618/65^2) (34,57385/67^2) (35,62432/69^2) (36,67767/71^2) 
        (37,73398/73^2) (38,79333/75^2) (39,85580/77^2) (40,92147/79^2) 
        (41,99042/81^2) (42,106273/83^2) (43,113848/85^2) (44,121775/87^2) 
        (45,130062/89^2) (46,138717/91^2) (47,147748/93^2) (48,157163/95^2) 
        (49,166970/97^2) (50,177177/99^2) 
      };
      \addlegendentry{}
      \addplot[thick,every mark/.append style={solid},
      mark=triangle*,mark phase=1,mark repeat=12,loosely dashed,green] plot coordinates {
        (5,342/9^2) (6,517/11^2) (7,748/13^2) (8,1043/15^2) 
        (9,1410/17^2) (10,1857/19^2) (11,2392/21^2) (12,3023/23^2) 
        (13,3758/25^2) (14,4605/27^2) (15,5572/29^2) (16,6667/31^2) 
        (17,7898/33^2) (18,9273/35^2) (19,10800/37^2) (20,12487/39^2) 
        (21,14342/41^2) (22,16373/43^2) (23,18588/45^2) (24,20995/47^2) 
        (25,23602/49^2) (26,26417/51^2) (27,29448/53^2) (28,32703/55^2) 
        (29,36190/57^2) (30,39917/59^2) (31,43892/61^2) (32,48123/63^2) 
        (33,52618/65^2) (34,57385/67^2) (35,62432/69^2) (36,67767/71^2) 
        (37,73398/73^2) (38,79333/75^2) (39,85580/77^2) (40,92147/79^2) 
        (41,99042/81^2) (42,106273/83^2) (43,113848/85^2) (44,121775/87^2) 
        (45,130062/89^2) (46,138717/91^2) (47,147748/93^2) (48,157163/95^2) 
        (49,166970/97^2) (50,177177/99^2) 
      };
      \addlegendentry{}

      \addlegendimage{empty legend}
      \addlegendentry{\Adivalgo}
      \addplot[thick,every mark/.append style={solid},
      mark=square*,mark phase=9,mark repeat=12,dashdotted,red] plot coordinates {
        (5,1082/10^2) (6,2670/18^2) (7,3415/21^2) (8,6859/32^2) 
        (9,8389/36^2) (10,14863/50^2) (11,17611/55^2) (12,28561/72^2) 
        (13,33058/78^2) (14,50273/98^2) (15,57150/105^2) (16,82669/128^2) 
        (17,92655/136^2) (18,128878/162^2) (19,142803/171^2) (20,192362/200^2) 
        (21,211149/210^2) (22,277061/242^2) (23,301742/253^2) (24,387227/288^2) 
        (25,418920/300^2) (26,527614/338^2) (27,567555/351^2) (28,703260/392^2) 
        (29,752758/406^2) (30,919725/450^2)
      };
      \addlegendentry{}
      \addplot[thick,every mark/.append style={solid},
      mark=square*,mark phase=5,mark repeat=12,dashdotted,blue] plot coordinates {
        (5,1284/9^2) (6,1711/11^2) (7,2196/13^2) (8,2739/15^2) 
        (9,3340/17^2) (10,3999/19^2) (11,4716/21^2) (12,5491/23^2) 
        (13,6324/25^2) (14,7215/27^2) (15,8164/29^2) (16,9171/31^2) 
        (17,10236/33^2) (18,11359/35^2) (19,12540/37^2) (20,13779/39^2) 
        (21,15076/41^2) (22,16431/43^2) (23,17844/45^2) (24,19315/47^2) 
        (25,20844/49^2) (26,22431/51^2) (27,24076/53^2) (28,25779/55^2) 
        (29,27540/57^2) (30,29359/59^2) (31,31236/61^2) (32,33171/63^2) 
        (33,35164/65^2) (34,37215/67^2) (35,39324/69^2) (36,41491/71^2) 
        (37,43716/73^2) (38,45999/75^2) (39,48340/77^2) (40,50739/79^2)
      };
      \addlegendentry{}
      \addplot[thick,every mark/.append style={solid},
      mark=square*,mark phase=1,mark repeat=12,dashdotted,green] plot coordinates {
        (5,4006/15^2) (6,6950/21^2) (7,11260/28^2) (8,17306/36^2) 
        (9,25487/45^2) (10,36253/55^2) (11,50097/66^2) (12,67555/78^2) 
        (13,89206/91^2) (14,115672/105^2) (15,147618/120^2) (16,185752/136^2)
        (17,230825/153^2) (18,283629/171^2) (19,345006/190^2) (20,415833/210^2) 
        (21,497028/231^2) (22,589564/253^2) (23,694455/276^2) (24,812730/300^2) 
        (25,945507/325^2) (26,1093928/351^2) (27,1259155/378^2)
      };
      \addlegendentry{}
    \end{axis}
  \end{tikzpicture}
  \caption{Number of arithmetic operations, for the adaptive variants,
    compared with the \divalgo algorithm (\textsc{2D})}
  \label{fig:basicopAdapt2D}
\end{figure*}
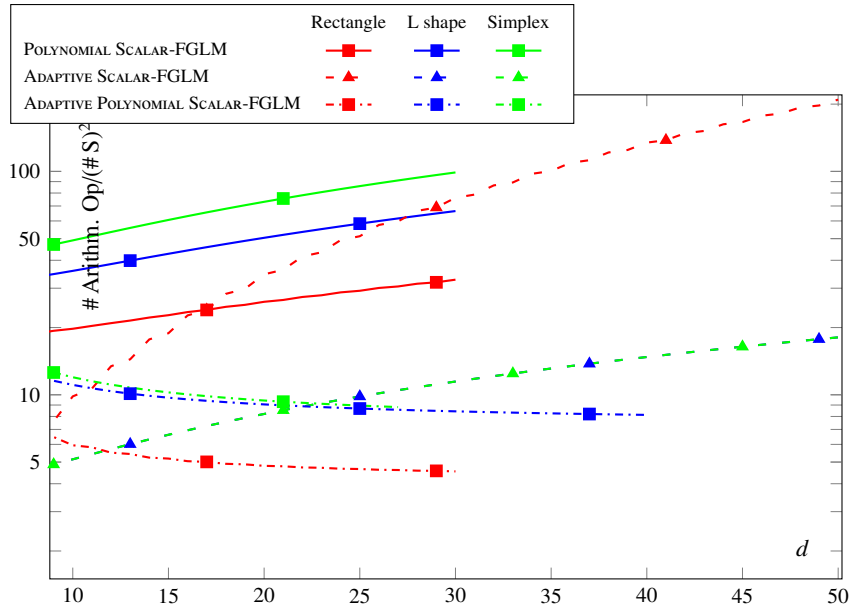
\begin{figure*}[htbp!]
  \centering
  \pgfplotsset{
    small,
    width=12cm,
    height=8cm,
    legend cell align=left,
    legend columns=4,
    legend style={at={(-0.05,0.94)},anchor=south
      west,font=\scriptsize,
    }
  }
  \begin{tikzpicture}[baseline]
    \begin{axis}[
      ymode=log,
      xlabel={$d$},
      xlabel style={at={(0.95,0.1)}},
      xmin=4.8,xmax=20.2,
      ymin=3.5,ymax=880,
      xtick={1,...,25},
      ytick={0.1,0.2,0.3,0.4,0.5,0.6,0.7,0.8,0.9,1,2,3,4,5,6,7,8,9,
        10,20,30,40,50,60,70,80,90,100,200,300,400,500,600,700,800,900,
        1000,2000,3000,4000,5000,6000,7000,8000,9000,10000,20000},
      yticklabels={},
      extra y ticks={0.1,0.5,1,5,10,50,100,500,1000,5000,10000},
      extra y tick labels={0.1,0.5,1,5,10,50,100,500,1000,5000,10000},
      ylabel={\#\,Arithm. Op/(\#\,S)$^2$},
      ylabel style={at={(0.08,0.75)}},
      ]
      \addlegendimage{empty legend}
      \addlegendentry{}
      \addlegendimage{legend image with text=Rectangle}
      \addlegendentry{}
      \addlegendimage{legend image with text=\textsc{L} shape}
      \addlegendentry{}
      \addlegendimage{legend image with text=Simplex}
      \addlegendentry{}

      \addlegendimage{empty legend}
      \addlegendentry{\divalgo}
      \addplot[thick,every mark/.append style={solid},
      mark=square*,mark phase=3,mark repeat=3,red] plot coordinates {
        (5,17547/20^2) (6,65196/36^2) (7,218286/63^2) (8,580542/96^2) 
        (9,799578/108^2) (10,3186694/200^2)
      };
      \addlegendentry{}
      \addplot[thick,every mark/.append style={solid},
      mark=square*,mark phase=2,mark repeat=3,blue] plot coordinates {
        (5,12131/13^2) (6,19919/16^2) (7,31000/19^2) (8,46210/22^2) 
        (9,66481/25^2) (10,92841/28^2) (11,126414/31^2) (12,168420/34^2) 
        (13,220175/37^2) (14,283091/40^2) (15,358676/43^2)
      };
      \addlegendentry{}
      \addplot[thick,every mark/.append style={solid},
      mark=square*,mark phase=1,mark repeat=3,green] plot coordinates {
        (5,201521/35^2) (6,616688/56^2) (7,1657736/84^2) (8,4019618/120^2) 
        (9,8950600/165^2) (10,18568195/220^2) (11,36288873/286^2)
        (12,67403361/364^2)
      };
      \addlegendentry{}

      \addlegendimage{empty legend}
      \addlegendentry{\asFGLM}
      \addplot[thick,every mark/.append style={solid},
      mark=triangle*,mark phase=3,mark repeat=3,loosely dashed,red] plot coordinates {
        (5,2132/20^2) (6,9923/36^2) (7,47318/63^2) (8,159544/96^2) 
        (9,225344/108^2) (10,1379580/200^2) (11,1831620/220^2) (12,4076069/288^2) 
        (13,10052631/390^2) (14,19865227/490^2) (15,24417922/525^2) (16,76100044/768^2) 
        (17,91250252/816^2) (18,154027835/972^2) (19,287286381/1197^2)
        (20,459286352/1400^2) 
      };
      \addlegendentry{}
      \addplot[thick,every mark/.append style={solid},
      mark=triangle*,mark phase=2,mark repeat=3,loosely dashed,blue] plot coordinates {
        (5,998/13^2) (6,1521/16^2) (7,2224/19^2) (8,3134/22^2) 
        (9,4278/25^2) (10,5683/28^2) (11,7376/31^2) (12,9384/34^2) 
        (13,11734/37^2) (14,14453/40^2) (15,17568/43^2) (16,21106/46^2) 
        (17,25094/49^2) (18,29559/52^2) (19,34528/55^2) (20,40028/58^2)
      };
      \addlegendentry{}
      \addplot[thick,every mark/.append style={solid},
      mark=triangle*,mark phase=1,mark repeat=3,loosely dashed,green] plot coordinates {
        (5,22575/35^2) (6,78848/56^2) (7,237160/84^2) (8,633100/120^2) 
        (9,1534225/165^2) (10,3433100/220^2) (11,7186608/286^2) (12,14216930/364^2) 
        (13,26792675/455^2) (14,48412000/560^2) (15,84313200/680^2) (16,142142168/816^2) 
        (17,232810325/969^2) (18,371581100/1140^2)
        (19,579427800/1330^2) (20,884710750/1540^2)
      };
      \addlegendentry{}

      \addlegendimage{empty legend}
      \addlegendentry{\Adivalgo}
      \addplot[thick,every mark/.append style={solid},
      mark=square*,mark phase=3,mark repeat=3,dashdotted,red] plot coordinates {
        (5,3958/20^2) (6,10622/36^2) (7,28635/63^2) (8,61500/96^2) 
        (9,76420/108^2) (10,240192/200^2) (11,287898/220^2) (12,480057/288^2)
      };
      \addlegendentry{}
      \addplot[thick,every mark/.append style={solid},
      mark=square*,mark phase=2,mark repeat=3,dashdotted,blue] plot coordinates {
        (5,3181/13^2) (6,4225/16^2) (7,5416/19^2) (8,6754/22^2) 
        (9,8239/25^2) (10,9871/28^2) (11,11650/31^2) (12,13576/34^2) 
        (13,15649/37^2) (14,17869/40^2) (15,20236/43^2) (16,22750/46^2) 
        (17,25411/49^2) (18,28219/52^2) (19,31174/55^2) (20,34276/58^2)
      };
      \addlegendentry{}
      \addplot[thick,every mark/.append style={solid},
      mark=square*,mark phase=1,mark repeat=3,dashdotted,green] plot coordinates {
        (5,38330/35^2) (6,91307/56^2) (7,193807/84^2) (8,377611/120^2)
      };
      \addlegendentry{}
    \end{axis}
  \end{tikzpicture}
  \caption{Number of arithmetic operations, for the adaptive variants,
    compared with the \divalgo algorithm (\textsc{3D})}
  \label{fig:basicopAdapt3D}
\end{figure*}
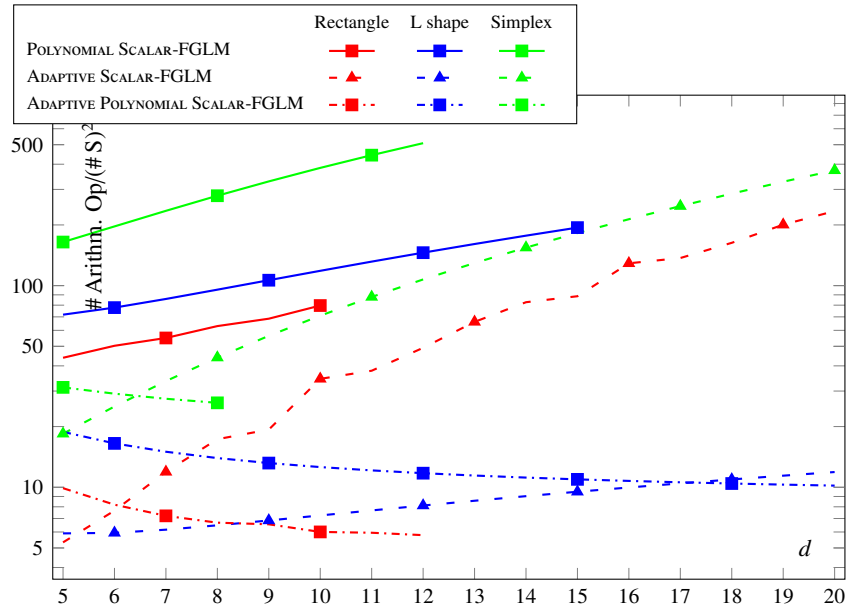
